\documentclass[12pt, final]{l4dc2023}

\usepackage{times}

\usepackage{graphics}
\usepackage{amsmath,amssymb,amsfonts}

\usepackage{algorithm}
\usepackage{algorithmic}
\usepackage{cleveref}
\usepackage{mathtools}
\usepackage{graphicx}
\usepackage{hhline}
\usepackage{multirow,array}
\usepackage{multirow,makecell}
\usepackage{pgfplots}
\usepackage{physics}
\usepackage{booktabs}
\usepackage{xurl}

\usepackage{multirow}
\usepgfplotslibrary{groupplots}
\usetikzlibrary{positioning}
\pgfplotsset{compat=newest}
\pgfplotsset{every axis legend/.append style={%
cells={anchor=west}}
}
\pgfplotsset{every axis/.append style={
                    label style={font=\footnotesize},
					tick label style={font=\footnotesize},
					legend style={font=\footnotesize}
                    }}
\Crefname{equation}{Eq.}{Eqs.}
\Crefname{figure}{Fig.}{Figs.}
\Crefname{tabular}{Tab.}{Tabs.}
\Crefname{section}{Sec.}{Secs.}

\newcommand{\data}{\mathcal{D}}
\newcommand{\poisoneddata}{{\tilde{\mathcal{D}}}}

\renewcommand{\norm}[1]{\left\lVert#1\right\rVert}
\newcommand{\frobenius}{\textrm{F}}

\DeclareMathOperator*{\argmin}{arg\,min}
\renewcommand{\trace}{Tr}
\title[Analysis and Detectability of Offline Data Poisoning Attacks on Linear Systems]{Analysis and Detectability of Offline Data Poisoning Attacks on Linear Dynamical Systems}
\author{%
 \Name{Alessio Russo} \Email{alessior@kth.se}\\
 \addr Division of Decision and Control Systems, KTH Royal Institute of Technology
}
\newtheorem{assumption}{Assumption}

\DeclareMathOperator{\vect}{vec}
\newcommand{\LS}{\textrm{LS}}

\newcommand{\arxiv}{}
\begin{document}

\maketitle

\begin{abstract}
In recent years, there has been a growing interest in the effects of data poisoning attacks on data-driven control methods. Poisoning attacks are well-known to the Machine Learning community, which, however, make use of assumptions, such as cross-sample independence, that in general do not hold for linear dynamical systems.
Consequently, these systems require different attack and detection methods than those developed for supervised learning problems in the i.i.d.\ setting.
Since most data-driven control algorithms make use of the least-squares estimator, we study how poisoning impacts the least-squares estimate through the lens of statistical testing, and question in what way data poisoning attacks can be detected.
We establish under which conditions the set of models compatible with the data includes the true model of the system, and we analyze different poisoning strategies for the attacker. On the basis of the arguments hereby presented, we propose a stealthy data poisoning attack on the least-squares estimator that can escape classical statistical tests, and conclude by showing the efficiency of the proposed attack. The code can be found here \url{https://github.com/rssalessio/data-poisoning-linear-systems}.
\end{abstract}

\section{Introduction}
Over the past few decades, the rise in computational power, data accessibility, technological progress, and successful results have fueled research in data-driven methods.
Nonetheless, these methods may be vulnerable to data poisoning attacks, which aim to degrade performance by altering the training data \cite{biggio2012poisoning,barreno2006can,barreno2010security}.
The concept of poisoning was originally introduced for anomaly detection \cite{barreno2006can,kloft2010online,rubinstein2009antidote} and attacks against SVM models \cite{biggio2012poisoning}. Thereafter, various machine learning models have been shown to be susceptible to poisoning attacks \cite{jagielski2018manipulating,xiao2015feature,zhang2020online,shafahi2018poison} (see also \cite{tian2022comprehensive} for a survey).
In the field of systems control, data-driven methods are used due to a variety of reasons, such as decreased modeling complexity and/or reduced costs. In fact, data-driven control methods allow the user to formulate a control law directly from the data, thus bypassing the need of modeling the dynamical systems. 
There are a variety of these methods to use:  techniques based on Willem's et al. lemma \cite{willems2005note,de2019formulas,coulson2019data}, Virtual Reference Feedback Tuning (VRFT) \cite{campi2002virtual}, Iterative Feedback Tuning \cite{hjalmarsson2002iterative}, Correlation-based Tuning \cite{karimi2004iterative}, etc.
However, data-driven control methods may also be affected by poisoning attacks. In  \cite{russo2021poisoning}, the authors formulate  a poisoning attack against the VRFT method. Similarly, in \cite{2209.09108} the authors propose a poisoning attack against data-driven predictive control methods.
In \cite{showkatbakhsh2016secure,feng2021learning}, the authors consider how to recover the underlying model of the system from poisoned data, while \cite{chekan2020regret} considers enlarging the confidence set of the poisoned parameters to improve the performance of online LQR. The detection and analysis of these attacks, however, have received limited attention.  
Through first principles thinking, we investigate which models are compatible with the poisoned data. We focus on the least-squares (LS) estimator and study how poisoning affects the performance of this estimator. Through the lens of statistical tests, we examine how data can be poisoned and how these attacks can be detected. Our last contribution is to propose an attack that can  impact the LS estimator while being stealthy to classical statistical tests, including  residual and correlation tests typically used for anomaly detection. We provide examples and numerical simulations to accompany our results.

\section{Related work}\label{sec:related_work}
Data poisoning attacks can be categorized into two classes: \textit{untargeted poisoning attacks} and \textit{targeted poisoning attacks} \cite{tian2022comprehensive}. Attacks in the former class lead to some form of denial-of-service and try to hinder the convergence of the target model. On the other hand, targeted attacks change the data so that the trained model behaves according to the goal of the attacker \cite{liu2017trojaning,shafahi2018poison}. 
Countermeasures include preventive measures (e.g., encryption) or reactive measures (e.g., detection). 
In \cite{nguyen2013exact,bhatia2017consistent}, the authors propose different techniques to recover a linear model from the data for oblivious adversaries, while for adaptive adversaries recovery is possible  under some stringent assumptions \cite{bhatia2015robust}. 
As mentioned in \cite{bhatia2017consistent}, it seems unlikely that consistent estimators are even possible in face of a fully adaptive adversary.
However, to the best of our knowledge, most of these techniques are not directly applicable to control problems, due to the underlying dynamics of the system.  In \cite{alfeld2016data}, the authors study a targeted attack that poisons the forecasting of an autoregressive model. In \cite{showkatbakhsh2016secure}, the authors consider the problem of identifying a system whose output measurements have been corrupted by an adversary. They consider an omniscient adversary and given a bound on the number of attacked sensors, and some observability conditions, it is possible to derive a model that is useful for stabilizing the original system. A similar problem is studied in \cite{feng2021learning}: the authors consider a linear system affected by an unknown sparse adversarial disturbance $d_t$, and study how to recover the original model of the system.
A different problem is studied in \cite{chekan2020regret}, where the authors consider online poisoning of the adaptive LQR method
\cite{abbasi2011regret}, and, to compensate for the attack, they enlarge the confidence set of the estimator.
In \cite{russo2021poisoning} they formulate  a  bi-level optimization problem  to compute poisoning attacks against data-driven control methods.  Their attack is then applied to the digital twin of a building to demonstrate the potential of their attack \cite{russo2021data}. Lastly, in \cite{2209.09108} they extend the bi-level attack problem to attack data-driven predictive control methods.

\section{Preliminaries}\label{sec:preliminaries}

\paragraph{Model.} We consider a discrete-time LTI system affected by process noise:
\begin{equation}\label{eq:system}
x_{t+1}=A_\star x_t+B_\star u_t+w_t,
\end{equation}
where $t \in \mathbb{Z}$ is the discrete time variable, $x_t\in \mathbb{R}^n$ is the state of the system, $u_t\in \mathbb{R}^m$ is the control signal, $A_\star\in\mathbb{R}^{n\times n}, B_\star\in \mathbb{R}^{n\times  m}$ are the unknown system matrices, and $w_t\in \mathcal{W}\subseteq\mathbb{R}^n$ is an unmeasured disturbance belonging to some convex set $\mathcal{W}$ (which can be the entire $\mathbb{R}^n$). 
For a sequence of input-state measurements $\{u_k\}_{k=0}^{T-1}, \{x_k\}_{k=0}^T$, we define $\psi_t = (x_t^\top, u_t^\top)^\top \in \mathbb{R}^{n+m}$ and the following data matrices:
\begin{align*}
X_+ \coloneqq \begin{bmatrix}x_1 & \dots & x_T\end{bmatrix}, \quad X_-\coloneqq \begin{bmatrix}x_0 & \dots & x_{T-1}\end{bmatrix},\quad
U_- \coloneqq \begin{bmatrix}u_0 & \dots & u_{T-1}\end{bmatrix}, 
\end{align*}
and $X = \begin{bmatrix}x_0 & \dots & x_{T}\end{bmatrix}, \Psi_- = \begin{bmatrix} \psi_0 &\dots & \psi_{T-1}\end{bmatrix}$. We make the assumption that the user has access to one trajectory of the system,  used for identification or data-driven control.
\begin{assumption}\label{assumption:rank_condition}
The data $\mathcal{D}=(U_-,X)$ available to the user consists of one input-state trajectory of length $T$. Furthermore, the data satisfy the rank condition $\rank \Psi_-=n+m$.
\end{assumption}
This is a standard assumption in data-driven control, and it can be guaranteed for noise-free systems by choosing a persistently exciting input signal of order $n+1$  \cite{willems2005note}.

\begin{figure}[t]
    \centering
    \includegraphics[width=0.8\linewidth]{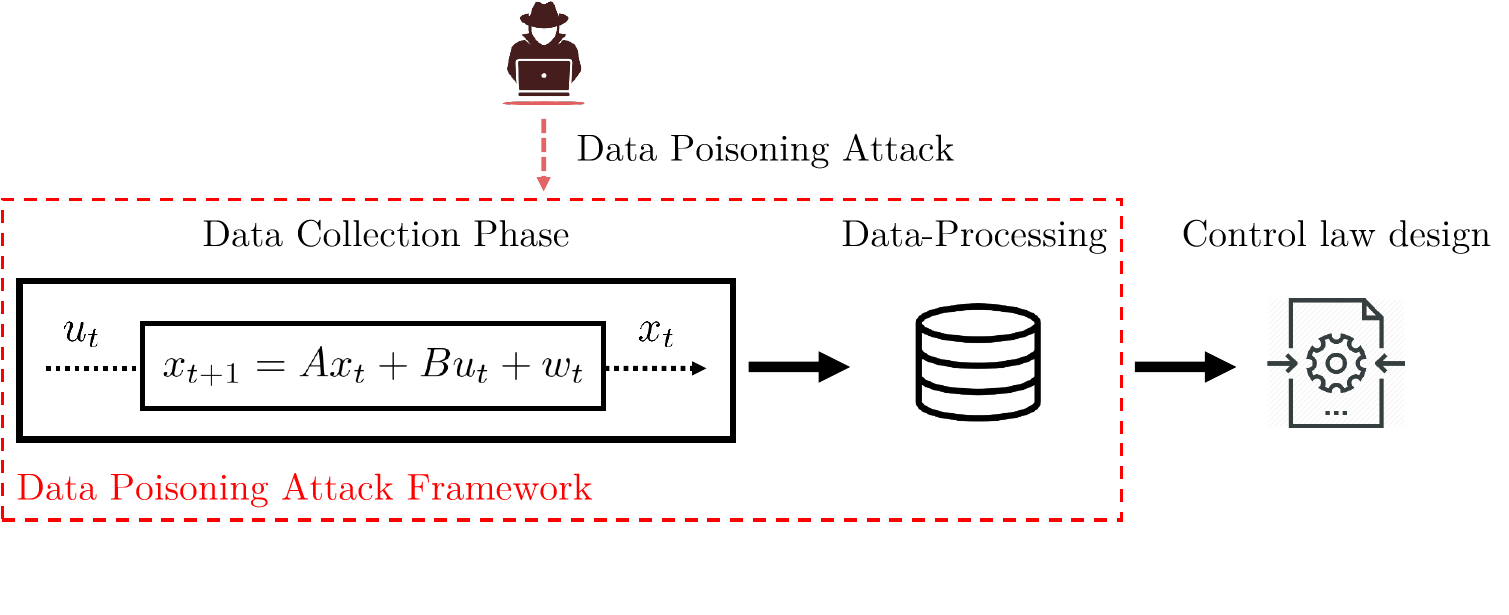}
    \caption{Data poisoning framework for data-driven control methods (figure adapted from \cite{russo2021poisoning}).}
    \label{fig:data_poisoning_attack_scheme}
\end{figure}

\paragraph{Data poisoning attacks.} We denote by $\Delta \data \coloneqq (\Delta U_-,\Delta X)$ the poisoning signals on the input-state measurements, so that $\Delta U_-=\begin{bmatrix}\Delta u_0 &\dots \Delta u_{T-1}\end{bmatrix} \in \mathbb{R}^{m\times T}$, $\Delta X= \begin{bmatrix}\Delta x_0 &\dots \Delta x_{T}\end{bmatrix}\in \mathbb{R}^{n\times T+1}$. We let $\tilde U_- =U_-+\Delta U_-$, $\tilde X=X+\Delta X$ be the resulting poisoned signals. Similarly, we denote the poisoned dataset by $\poisoneddata \coloneqq (\tilde U_-,\tilde X)$ and let $\tilde \psi_t=\psi_t+\Delta \psi_t$,  where $\Delta \psi_t = (\Delta x_t^\top, \Delta u_t^\top)^\top$ (sim. we define $\tilde \Psi_-$, and $\Delta \Psi_-$).  
Attacks in the literature  are generally classified as \textit{targeted attacks} or \textit{untargeted attacks}. \textit{Untargeted attacks} just try to alter the performance of the data-driven control scheme, causing a denial-of-service. 
On the other hand,  \textit{targeted attacks} are usually carried out by the attacker  to achieve some specific goals, e.g., making the closed-loop system unstable, maximizing the energy used by the system, making the system uncontrollable, etc. The attacker's goal is formulated as a  bi-level optimization problem (see also \cite{russo2021poisoning})
    \begin{equation}
        \begin{aligned}
            \max_{(\Delta U_{-},\Delta X)\in \mathcal{U}\times\mathcal{X}} \quad & \mathcal{A}\left(\data, K\right)           \textrm{ s.t. }   K \in \argmin_{K'} \mathcal{L}(U_-+\Delta U_-, X+\Delta X,K'),
        \end{aligned}
    \end{equation}
    where $(\mathcal{U},\mathcal{X})$ is a convex set, $K$ is the closed-loop controller, $\mathcal{A}$ represents the objective function of the malicious agent, and $\mathcal{L}$ represents the function used by the victim to compute the control law $K$ according to the poisoned data $(\tilde U_-, \tilde  X)$. 

\section{Attacks and Detection Strategies}
In this section we examine what is the set of pairs $(A,B)$ that are compatible with  the poisoned data $\poisoneddata$, and establish a sufficient and necessary condition for $(A_\star,B_\star)$ to be compatible with $\poisoneddata$. We investigate  how least-squares estimate changes under poisoning, examine attack detection, and propose a stealthy untargeted attack on the least-squares estimate.
\subsection{The set of compatible models under data poisoning}
Most data-driven control methods assume that there exists a linear system $(A,B)$ that is consistent with the data. Ignoring the noise term $w_t$ in \cref{eq:system}, consistency amounts to finding all $(A,B)\in \mathbb{R}^{n\times (n+m)}$ that satisfy the equation $X_+=AX_-+BU_-$ (\emph{i.e.}, all the pairs consistent with the dataset). In presence of a noise signal $w_t$, consistency is derived from the following relationship
\begin{equation}\label{eq:fundamental_relationship}
    \begin{bmatrix}
    I_n & -A_\star & -B_\star
    \end{bmatrix} \begin{bmatrix}
    X_+\\ \Psi_-
    \end{bmatrix}=W_-,
\end{equation}
where $W_- =\begin{bmatrix}
w_0&w_1&\dots &w_{T-1}
\end{bmatrix}\in\mathcal{W}^T$ (where $\mathcal{W}^T\coloneqq\bigtimes_{n=1}^{T} \mathcal{W}$). Then, given a poisoned dataset $\tilde\data=(\tilde U_-,\tilde X)$,  we wonder for which $(A,B)\in \mathbb{R}^{n\times (n+m)}$ and $\tilde W\in \mathcal{W}^T$ the following relationship holds
\begin{equation}\label{eq:fundamental_eq_poisoned_system_simplified}
    \begin{bmatrix}
    I_n & -A & -B
    \end{bmatrix} \begin{bmatrix}
    \tilde X_+ \\ \tilde \Psi_-
    \end{bmatrix}= \tilde W_-.
\end{equation}
To answer this question, we seek the  set of noise sequences  $\mathcal{W}_\poisoneddata$ that are compatible with  the data $\poisoneddata$. 
Following a similar approach as in \cite{koch2020verifying}, we note that the compatible noise terms $\{\tilde w_t\}_{t}$ belong to the image of $\tilde G\coloneqq \begin{bmatrix}
\tilde X_+^\top & \tilde \Psi_-^\top
\end{bmatrix}^\top$.  Straightforwardly, if $\tilde W_-\tilde G^\perp=0$, then $\tilde W_-$ is compatible with the data (where $\tilde G^\perp$ denotes a basis of the kernel of $\tilde G$), and $\mathcal{W}_{\poisoneddata}$ is
\begin{equation}\mathcal{W}_{\poisoneddata} =\left\{
 \tilde W_- \in \mathcal{W}^T: \tilde W_-\tilde G^\perp =0 \right\}.
 \end{equation}
Therefore, the following result  characterizes in which cases $(A_\star,B_\star)\in \Sigma_{\poisoneddata}$.
\begin{lemma}\label{lemma:compatible_a_b}\footnote{\ifdefined \arxiv All the proofs can be found in the appendix. \else Refer to the technical report \url{https://arxiv.org/abs/2211.08804} for all the proofs and numerical details.\fi}
Consider a poisoned dataset $\poisoneddata$. The set of all pairs $(A,B)$ that are consistent with the data $\poisoneddata$ is
$
    \Sigma_{\tilde\data}\coloneqq \left\{(A,B)\in \mathbb{R}^{n\times n+m}:
    \exists \tilde W_-\in \mathcal{W}_\poisoneddata: (\ref{eq:fundamental_eq_poisoned_system_simplified}) \hbox{ holds for } (A,B, \tilde W_-)  \right\}.
$
Let $\tilde W_\star = W_-+ \Delta W_-$, where $\Delta W_-=\Delta X_+ -A_\star \Delta X_- - B_\star \Delta U_-$. Then  $(A_\star, B_\star) \in \Sigma_{\poisoneddata}\iff \tilde W_\star \in \mathcal{W}^T.$
 \end{lemma}

\noindent In most applications, $\mathcal{W}$ is considered to be $\mathbb{R}^n$ itself, which essentially guarantees that $(A_\star,B_\star) \in \Sigma_\poisoneddata$. However, if the disturbance $\{w_t\}_{t\geq0}$ is generated according to some probability measure $\mathbb{P}$, then the likelihood of $\tilde W_\star$ may be  small under $\mathbb{P}$ depending on the attack.
In other applications, $\mathcal{W}$ is a known bounded convex set, and may not contain the sequence $\tilde W_\star$. In this case, it is necessary to enlarge $\mathcal{W}$ to be able to recover the original model. Lastly, in some scenarios the user may know in advance an over-approximate $\hat{\Sigma}_\data$ of $\Sigma_\data$, which can be used to infer if the data has been poisoned in case $\hat{\Sigma}_\data$ and $\Sigma_\poisoneddata$ are too different (using, for example, Bayesian hypothesis testing).

This result can also be interpreted in the following way: if $\mathcal{W}$ is a bounded convex set, an attacker may try to bound $\Delta X$ and $\Delta U_-$ to bound $\Delta W_-$, and, consequently, act on the compatibility of the true system matrices. Obviously, upper bounding the norm of the poisoning signals seems like a good way to make sure that  is less detectable. However, just bounding the norm of the poisoning signals, as we see in the forthcoming sections, is not enough to achieve undetectability of an attack.

\subsection{Attack strategies for the least-squares estimator}\label{subsec:lsestimator}
To better understand how to formulate attack strategies, and analyze the problem of detectability, we now turn our attention to the least-squares (LS) estimator. As pointed out in \cite{de2019formulas}, this LS estimate is used in the formulation of data-driven controllers. The unpoisoned LS estimate  can be compactly written as $( A_\LS,  B_\LS) = X_+ \Psi_-^\dagger$ ($\Psi_-^\dagger$ is the right inverse of $\Psi_-$). We also denote by $(A_\LS, B_\LS)$ and $(\tilde A_\LS,\tilde B_\LS)$ the LS estimates when, respectively,  $\data$  and $\poisoneddata$ are used. Furthermore, let $(\Delta \tilde A_\LS,\Delta \tilde B_\LS)=(\tilde A_\LS - A_\star,\tilde B_\LS- B_\star)$ the difference between the LS estimate and the true parameter, and indicate by $\Delta \tilde \theta_\LS = \vect(\Delta \tilde A_\LS,\Delta \tilde B_\LS)$ its vectorization.  We further assume that $\rank \tilde \Psi_-=n+m$.  In presence of a poisoning attack $(\Delta X, \Delta U_-)$ we obtain the following straightforward result, which is used to discuss possible attack strategies.
\begin{lemma}\label{lemma:estimate_delta}
The LS  error is given by
$
    (\Delta \tilde A_\LS, \Delta \tilde B_\LS) =  \tilde W_\star \tilde \Psi_-^\dagger,
$
where  $\tilde W_\star$ is as in lemma \ref{lemma:compatible_a_b}. In addition to that, we have
$
\sigma_{\text{min}}(W_-) + \sigma_{\text{min}}(\Delta W_-)\leq \|\tilde \Psi_-\|_{\textnormal{F}} \|\Delta \tilde \theta_\LS\|_2 \leq  \sigma_{\text{max}}(W_-) + \sigma_{\text{max}}(\Delta W_-)$, where $\|\cdot\|_\textnormal{F}$ indicates the Frobenius norm and $\sigma_{\text{min}}, \sigma_{\text{max}}$ the minimum and maximum singular values.
\end{lemma}
\paragraph{Relationship between poisoning and exploration.} 
This result provides a way to formulate possible attack strategies and to analyze their impact. The adversary clearly needs to minimize the amount of exploration, quantified by the term  $\|\tilde \Psi_-\|_{\textnormal{F}}=\sum_{t=0}^{T-1}\|\tilde \psi_t\|_2 $ to maximize the error.  Fundamentally, any attack wishing to maximize the error needs to change the data as to minimize the exploration performed by the victim.  As an informal argument, define $\tilde C_T = \sum_{t=0}^{T-1}(\tilde\psi_t \otimes I_n) (\tilde\psi_t\otimes I_n)^\top$ and introduce the {\it unexcitation subspace}  $\tilde{\mathcal{U}}=\left\{\theta \in \mathbb{R}^{n(n+m)}:  \lim\sup_{T\to\infty} \theta^\top \tilde C_T\theta < \infty\right\}$ (see \cite{bittanti1992effective}), and let $\tilde{\mathcal{E}}$ be its orthogonal complement (the data generation process is undefined on purpose, since it is just an illustrative argument).
Denote by $\Delta\tilde\theta_\LS^{\tilde{\mathcal{E}}}$ and $\Delta\tilde \theta_\LS^{\tilde{\mathcal{U}}}$ the orthogonal projections of $\Delta\tilde \theta_\LS$ on these two subspaces, so that $\Delta\tilde \theta_\LS=\Delta\tilde \theta_\LS^{\tilde{\mathcal{E}}} + \Delta\tilde\theta_\LS^{\tilde{\mathcal{U}}}$.  Then, under some simple assumptions, it is possible to show that asymptotically $\|\Delta \tilde\theta_\LS\|_2 = \|\Delta\tilde \theta_\LS^{\tilde{\mathcal{U}}}\|_2$. 
Hence, maximizing $\|\Delta\tilde \theta_\LS\|$ amounts to changing the unpoisoned regressor $ \psi_t$ so that,  the unexcitation subspace becomes "\emph{larger}", which implies that the amount of exploration is  lowered.  
\ifdefined\arxivTo
To see this, note that  $\|\Delta \tilde \theta_\LS\|_{\tilde C_T^{1/2}} \geq  \sqrt{(\Delta\tilde\theta_\LS^{\tilde{\mathcal{E}}})^\top \tilde C_T  \Delta\tilde\theta_\LS^{\tilde{\mathcal{E}}}}  - \sqrt{(\Delta\tilde\theta_\LS^{\tilde{\mathcal{U}}})^\top \tilde C_T  \Delta\tilde\theta_\LS^{\tilde{\mathcal{U}}}}$, where the second term is clearly bounded for all $T$, while the first one diverges as $T\to\infty$ unless $\Delta\tilde\theta_\LS^{\tilde{\mathcal{E}}}=0$ asymptotically. 
If the LS estimator converges, we conclude that asymptotically  $\Delta\tilde\theta_\LS^{\tilde{\mathcal{E}}}=0 \Rightarrow \|\Delta \tilde\theta_\LS\|_2 = \|\Delta\tilde\theta_\LS^{\tilde{\mathcal{U}}}\|_2 $. 
\else\fi
 To formalize the concept, let $\{v_i\}_{i=1}^{n(n+m)}$ be an orthonormal basis of $\mathbb{R}^{n(n+m)}$ with $v_1 = \theta_\LS/ \|\theta_\LS\|_2$, where  $\theta_\LS=\vect(A_\LS, B_\LS)$ is the true LS-estimate for an unpoisoned dataset. Then, the estimation error in the direction of $v_k$ is given by $|v_k^\top\Delta \tilde \theta_\LS| $, which is lower bounded as follows.
\begin{corollary}\label{corollary:lower_bound_error_exploration}
For any $k =1,\dots,n(n+m)$, and a poisoned dataset $\poisoneddata$
\begin{equation}
\sqrt{( v_k^\top \Delta \tilde \theta_\LS)^2} \geq \frac{|\cos(\alpha_k)| (\sigma_{\text{min}}(W_-) +\sigma_{\text{min}}(\Delta W_-))}{\|V_k \tilde\Psi_-\|_{\textnormal{F}}},
\end{equation}
where $V_k = \vect_{n,n+m}^{-1}(v_k)$\footnote{$\vect_{ab}^{-1}(x)$ reshapes a vector $x\in\mathbb{R}^{ab}$ into a matrix of size $a\times b$ by arranging the elements of $x$ column-wise.}, and $\alpha_k$ is the angle between $\vect(V_k \tilde \Psi_-)$ and $\vect(\tilde W_\star)$.
\end{corollary}
The term $\|V_k \tilde \Psi_-\|_\textnormal{F}$ can be interpreted as the  total amount of exploration in the direction of $v_k$. Minimizing the exploration in the direction of the true estimate $\theta_\LS$ implies a larger value of $|v_1^\top\Delta \tilde \theta_\LS| $, which is larger when the estimator error $\Delta \tilde \theta_\LS$ is parallel to $\theta_\LS$, from which we deduce that asymptotically the unexcitation subspace includes the unpoisoned estimate  $\theta_\LS$. 

These results not only shed a light on the mechanics of poisoning, but also help us define a possible  attack. The attacker can compute some poisoning signals $(\Delta X,\Delta U_-)$ by solving the convex problem $\min_{(\Delta X,\Delta U_-)\in\mathcal{C}} \|\tilde \Psi_-\|_{\textnormal{F}}$ (or $\min_{(\Delta X,\Delta U_-)\in\mathcal{C}} \|V_1 \tilde \Psi_-\|_{\textnormal{F}}$), where $\mathcal{C}$ is some convex set.
Nevertheless, this simple attack ignores other terms, such as $\Delta W_-$, and therefore may not be enough to significantly impact the LS estimate.  As we discuss in the next sections, maximizing the norm of the LS residuals fills this gap. Furthermore, as we see,  an attacker needs to impose additional constraints on the optimization problem to make an attack stealthy. To that aim, we begin by discussing how statistical hypothesis testing can help to detect poisoning attacks.

\subsection{Detection  analysis for the LS estimator}
The detection of any attack should be based on two important ingredients: (1)  prior knowledge of the system and its signals; (2) independent statistical tests. Prior knowledge is useful to detect possible wrongdoings, however, that knowledge may be biased. Therefore, it is important to complement tests based on prior knowledge with tests that are independent of that  knowledge. 
\ifdefined\arxiv
Prior knowledge, for example, includes confidence bounds on the variance of the noise and/or knowledge of the statistical properties of $\{(x_t,u_t)\}_t$. If any of those are known, it is possible to derive one-sample tests to assess these statistics.\else\fi
In the following, we relate poisoning attacks to classical statistical tests.
We begin our study by considering attacks on the input signal, and then consider attacks on the state signal as well.

\subsubsection{Detection of Attacks on the input signal} The statistical properties of the input signal are usually assumed to be known. In fact, in most experiments, the input is usually chosen as a white noise signal, as to excite the dynamics of the system. Assuming $\{u_t\}_t$ is a sequence of i.i.d. random variables with distribution $P$, whiteness tests \cite{box2015time,drouiche2000new} can be used to deduce if the samples in  $\{\tilde u_t\}_t$ are white, while one-sample tests (such as the Kolmogorov-Smirnov test \cite{massey1951kolmogorov}, or the Anderson-Darling test \cite{nelson1998anderson}) can be used to assess whether the samples in  $\{\tilde u_t\}_t$ are distributed according to $P$. More simply, if $u_t \sim \mathcal{N}(0,I_m)$, then $\|U_-\|_\textnormal{F}^2 \sim \chi^2(Tm)$ is a Chi-squared distribution with $Tm$ degrees of freedom. Consequently, for a small $\delta\geq 0$ we see the constraint $\|\Delta U_-\|_\textnormal{F}\leq \delta  \|U_-\|_\textnormal{F}$ as a way to constraint the Chi-squared statistics.
Along this reasoning, an important class of input attacks can be derived when $(u_t,\tilde u_t)$ are indistinguishable, i.e., statistically equivalent.
\begin{definition}
Suppose that $u_t$ is i.i.d., distributed according to $P$ for every $t\geq 0$. Then, $(u_t,\tilde u_t)$ are indistinguishable if $\tilde u_t$ is i.i.d. and distributed according to $P$  for every $ t\geq 0$.
\end{definition}
To illustrate the attack, consider the following example.
\begin{example}\label{example:input_poisoning}
Consider the  system $x_{t+1}=a_\star x_t +b_\star u_t+w_t$, with $(a_\star,b_\star)=(0.7,0.5)$ and $w_t\sim\mathcal{N}(0, 1)$.   In \cref{fig:input_poisoning_example_zonotope} are shown the confidence intervals for the LS estimate when the input has been poisoned according to the indistinguishable attack $\tilde u_t$, where $(u_t,\tilde u_t)$ are independent, i.i.d., and distributed according to $\mathcal{N}(0,1)$. Under poisoning, the estimate of $b_\star$ converges to $0$ as $T\to\infty$.
\end{example}
\begin{figure}[t]
    \centering
    \includegraphics[width=.8\linewidth]{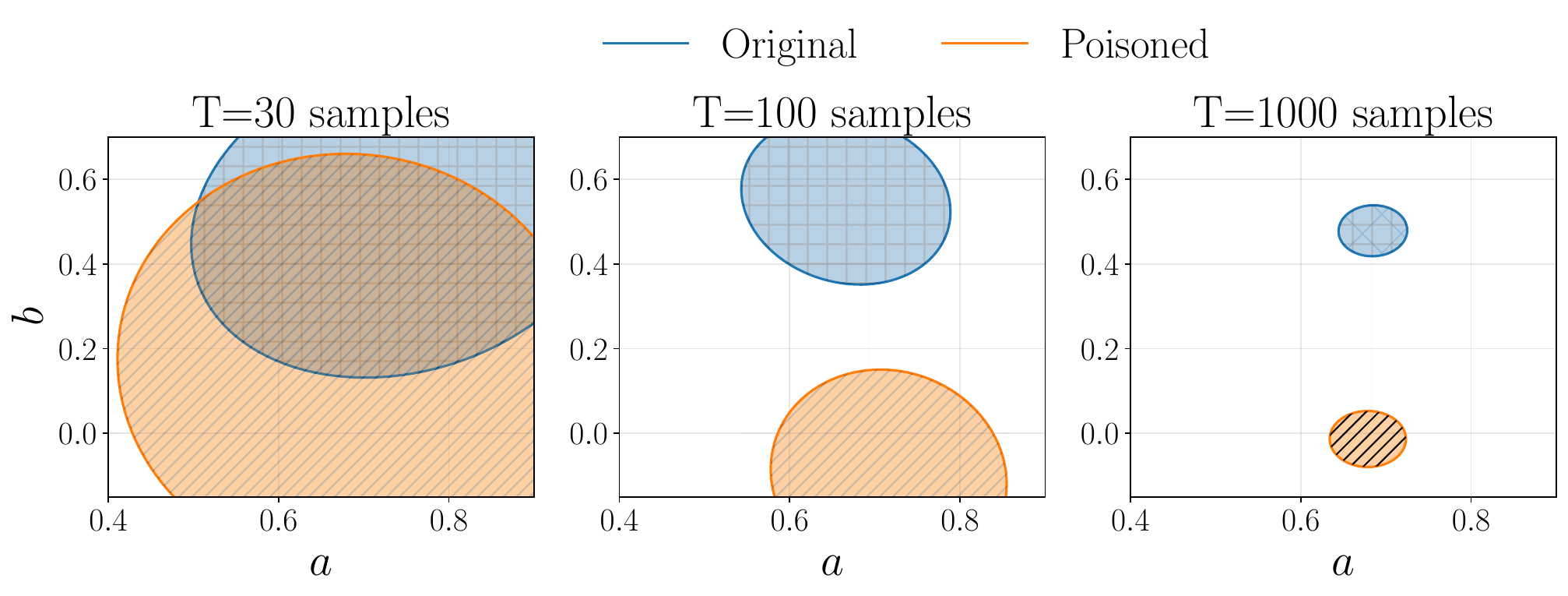}
    \caption{Example of input poisoning (see \cref{example:input_poisoning}). The colored ellipses depict the 95\% confidence interval of the LS estimates of the true parameters $(a_\star,b_\star)$.  When the input data is  poisoned, we obtain an $F$-statistic of $Z_\poisoneddata\approx(0.45, 0.62, 0.15)$ respectively for $T=(15,100,1000)$ samples. Otherwise, we obtain  $Z_\data \approx (7.35,29.05,233.35)$.}
    \label{fig:input_poisoning_example_zonotope}
\end{figure}
This insight leads  us to  the following result.
\begin{lemma}\label{lemma:input_attack_indistinguishable_asymptotics}
Let $\{u_t\}_{t=0}^{T-1}$ be an i.i.d. sequence distributed according to $\mathcal{N}(0,\Sigma_u)$. Assume $\Delta u_t=-u_t+a_t$ to be an indistinguishable attack, with $a_t\sim \mathcal{N}(0,\Sigma_u)$ independent of $u_t$, and $\Delta x_t=0, t\geq 0$. Then, if $A$ is stable,  $\Delta \tilde B_\LS\to -B_\star$ w.p. $1$ as $T\to \infty$.
\end{lemma}
Intuitively, if the poisoned input is completely uncorrelated from the data, then the best estimate  of $B_\star$ is $0$. Simply, as the number of samples grows larger,  $B_\star+\Delta \tilde B_\LS \to 0$.

To detect this type of attack, we propose to test the explanatory power of the input data. 
Using classical partial $F$-tests \cite{kleinbaum2013applied}, we test the hypothesis $H_0: \|B_\star\|=0$ against $H_1: \|B_\star\|\neq 0$.  Assume the underlying system is affected by some process noise $w_t \sim \mathcal{N}(0,\Sigma_w)$. Denote by $\tilde A^{(1)}_\LS$ the LS estimate of $A_\star$  when the input signal $\tilde u_t$ is not used by the LS estimator. Similarly, denote by $( \tilde A^{(2)}_\LS, \tilde B^{(2)}_\LS)$ the LS estimate when  $\tilde u_t$ is considered in the estimation process.
Consider the LS residuals, and define the statistic
\begin{equation}
Z_\poisoneddata\coloneqq\frac{\norm{X_+ -  \tilde A^{(1)}_\LS X_-}_\textnormal{F}^2 -\norm{X_+ -  \tilde A^{(2)}_\LS X_-+  \tilde B^{(2)}_\LS \tilde U_-}_\textnormal{F}^2 }{(nm/(T-n(n+m)-1))\norm{X_+- \tilde A^{(2)}_\LS X_-+  \tilde B^{(2)}_\LS \tilde U_-}_\textnormal{F}^2 }.
\end{equation}
Under $H_0$ it can be  shown that the statistics $Z_\poisoneddata$ follows an $F$ distribution with $(nm,T-n(n+m)-1)$ degrees of freedom (follows from an  application of \cite[Lemma II.4]{ljung1998system}). Using this partial $F$-test, we reject $H_0$ if and only if $Z_\poisoneddata> f_{nm,T-n(n+m)-1}^\alpha$, where $f_{a,b}^\alpha$ is  the upper $\alpha$-point of an $F$ distribution with $(a,b)$ degrees of freedom. In conclusion, not rejecting $H_0$ may indicate that the input data has been poisoned. Clearly, for more complex cases, we need to resort to other tools, such as the analysis of the residuals, as explained in the following section.


\subsubsection{Residual analysis}\label{sec:residual_analysis}
We claim that any  attacker that wishes to remain stealthy needs to make sure that the residuals of the LS procedure satisfy certain statistical conditions. We begin by deriving the following bound on the residuals of the LS estimate for generic attacks that are independent of the noise signal $\{w_t\}_t$ (this includes the class of oblivious attacks). Let the residual of $(\tilde A_\LS, \tilde B_\LS)$ at time $t$ be $\tilde R_t= \tilde x_{t+1}-\tilde A_\LS\tilde x_t - \tilde B_\LS\tilde u_t$. In matrix notation, we write $\tilde R = \begin{bmatrix}
 \tilde R_0 &\tilde R_1 &\dots &\tilde R_{T-1}
\end{bmatrix}=\tilde X_+ - \begin{bmatrix}\tilde A_\LS & \tilde B_\LS\end{bmatrix}\tilde \Psi_-$ (similarly, we denote by $R$ the residuals in absence of poisoning).

\begin{lemma}\label{lemma:MSE_residuals_under_attack}
Assume  $\{(\Delta x_t,\Delta u_t)\}_t$  to be independent of the i.i.d. noise sequence $\{w_t\}_t$, with $w_t\sim \mathcal{N}(0,\Sigma_w)$.  Then, the MSE $\mathbb{E}\left[\|\tilde R\|_{\textnormal{F}}^2\right]$ satisfies
$
    0\leq\mathbb{E}\left[\|\tilde R\|_{\textnormal{F}}^2-\norm{ R}_{\textnormal{F}}^2\right] \leq  \mathbb{E}[\sum_{i}\sigma_i^2(\Delta W_-)]
$, where $\sigma_i$ is the i-th singular value.
  Furthermore, $\|  R\|_{\textnormal{F}}^2$ is a quadratic form of a normal random vector, distributed according to $\sum_{i=1}^n \lambda_i \chi^2(T-n-m)$, with $\lambda_i$ being the $i$-th eigenvalue of $\Sigma_w$.
\end{lemma}
\ifdefined\arxiv As a corollary, for  oblivious random attacks, we find that if $\Delta x_t \sim \mathcal{N}(0,\Sigma_{\Delta x})$ and $\Delta u_t \sim \mathcal{N}(0,\Sigma_{\Delta u})$ are i.i.d., then the previous result can be improved as follows
$
    \mathbb{E}\left[\|\tilde R\|_{\textnormal{F}}^2-\| R\|_{\textnormal{F}}^2\right] = (T-n-m) \left[\trace\left(\Sigma_{\Delta x}+A_\star\Sigma_{\Delta x}A_\star^\top + B_\star\Sigma_{\Delta u}B_\star^\top\right)\right],
$
with $\|\tilde R\|_{\textnormal{F}}^2\sim \sum_{i=1}^n \bar\lambda_i \chi^2(T-n-m)$, and  $\bar \lambda_i$ is the i-th eigenvalue of $\Sigma_{\Delta x}+A_\star\Sigma_{\Delta x}A_\star^\top + B_\star\Sigma_{\Delta u}B_\star^\top +\Sigma_w$.  This result also applies to the class of indistinguishable input attacks previously explained  and motivates why in \cref{example:input_poisoning}  the confidence region in the poisoned case is significantly larger than in the unpoisoned one (due to the larger variance). \else\fi
In addition, observe the following  lemma on the sensitivity of the residuals.
\begin{lemma}[Sensitivity]\label{lemma:attack_sensitivity}
For any fixed attack satisfying the rank condition $\rank \tilde \Psi_-=n+m$, we obtain the following sensitivity on the residuals $\frac{\|\tilde R - R\|_{\textnormal{F}}}{\|R\|_{\textnormal{F}}} \leq \|\Delta \Psi_-\|_2 \|\tilde \Psi_-\|_{\textnormal{F}}^{-1}$.
\end{lemma}
Lemma \ref{lemma:MSE_residuals_under_attack} and \ref{lemma:attack_sensitivity} link the problem of maximizing the LS error to that of minimizing the amount of exploration $\|\tilde \Psi_-\|_{\textnormal{F}}$ (as discussed in \cref{subsec:lsestimator}) as well as maximizing the singular values of $\Delta W_-$. Furthermore,  Lemma \ref{lemma:MSE_residuals_under_attack} can be used to formulate a possible detection test on the variance of the residuals. 
In fact, we note that any attack independent of the noise will necessarily increase the variance of the residuals. 
This observation provides us a hint to test the variance of the residuals. It is possible to derive a two-tail test on the variance of the residuals (as long as the user knows has some knowledge on the covariance of the noise) to verify that the data has not been poisoned. 
The user can test whether $ \|\tilde R\|_{\textnormal{F}}^2$ belongs to the range
$
    Q_{(1-\alpha)/2, T-n-M}(\lambda_1,\dots,\lambda_n)\leq  \|\tilde R\|_{\textnormal{F}}^2 \leq Q_{\alpha/2, T-n-M}(\lambda_1,\dots,\lambda_n),
$
where $Q_{x,d}(\lambda_1,\dots)$ is the critical value of the distribution $\sum_{i=1}^n \lambda_i \chi^2(d)$ with significance $x\in(0,1)$. \ifdefined \arxiv 
If the values of $\lambda_1,\dots,\lambda_n$ are known to belong to some confidence region, the user can perform a  Bayesian likelihood ratio test.
\fi
Before we continue with an example, consider that the assumption of independence between $(\Delta U_-,\Delta X)$ and $W_-$ may not be always satisfied: if the attacker has access to the dataset $\data$, then it is  likely that she uses $X$, which depends on $W_-$, to compute the attack vector $(\Delta U_-,\Delta X)$. In other cases, for example, when the attacker has limited capabilities on the dataset and/or the poisoned sensors, the assumption of independence is more likely to hold. Similarly, the assumption holds whenever the attacker is executing an attack that has been computed on a different set of data. Moreover, from  simulations, it seems that this test is still valid to detect a possible adaptive attack.
\begin{example}[Untargeted attack]\label{example:residuals_variance}
Consider an attacker that  maximizes the norm of the residuals $\|\tilde R\|_\frobenius$ as a proxy to maximize $\norm{\begin{bmatrix}\Delta \tilde A_\LS&\Delta \tilde B_\LS\end{bmatrix}}_2$. Let $\delta \geq 0$ be a parameter that limits the amplitude of the poisoning signals, and define $\Delta \tilde W_-\coloneqq \Delta X_+ - A_\LS \Delta X_-- B_\LS \Delta U_-$. Then,  as detailed in the appendix, the adversary can solve the following concave problem to compute an attack:
\begin{equation}\label{eq:untargeted_attack_maximize_mse_error}
        \begin{aligned}
            \max_{\Delta U_{-}, \Delta X} \quad & \trace \left(( \Delta \Tilde W_-+2R)\Delta \Tilde W_-^\top\right)\quad
            \textrm{s.t.} \quad  \|\Delta X\|_{\textnormal{F}} \leq \delta \| X\|_{\textnormal{F}}, \|\Delta U_-\|_{\textnormal{F}} \leq \delta \| U_-\|_{\textnormal{F}}.
        \end{aligned}
\end{equation}
Consider the $4$-dimensional system used in \cite{russo2021poisoning}, affected by white noise with standard deviation $\sigma_w=0.1$. The victim has collected $T=200$ samples using  $u_t\sim\mathcal{N}(0,1)$. In \cref{fig:untargeted_attack_residuals_variance} are shown the results when $(\ref{eq:untargeted_attack_maximize_mse_error})$ is solved using (1) convex-concave programming (CCP), (2) the cross-entropy method \cite{de2005tutorial} and (3) random sampling from a Gaussian distribution (check the appendix for details). Note that the attacks can be easily detected for small values of $\delta$ using the test on the residuals (central plot). On the right plot, observe that, as discussed after corollary \ref{corollary:lower_bound_error_exploration}, the vectors $\theta_\LS$ and $\Delta \tilde \theta_\LS$ tend to align with each other when the attack is impactful.
\end{example}
\begin{figure}[t]
    \centering
    \includegraphics[width=\linewidth]{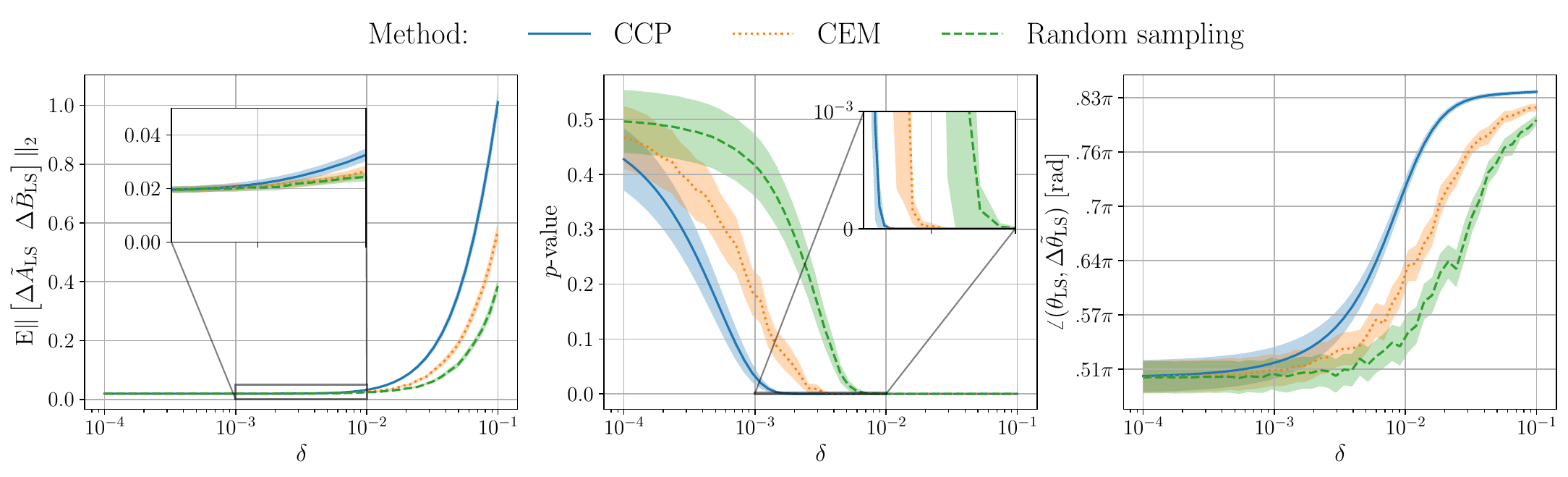}
    \caption{Untargeted attack that maximizes the MSE. The curves indicate an average and its 95\% confidence interval over 100 runs. From left to right : (1)  impact on the LS estimate; (2)  the p-value (right-tail) under the assumption of white noise with $\sigma_w=0.1$; (3) the angle between the unpoisoned estimate $\theta_\LS$ and the  error of the poisoned estimate $\Delta\tilde\theta_\LS$.}
    \label{fig:untargeted_attack_residuals_variance}
\end{figure}
This last example indicates that minimizing the norm of the poisoning signal is not enough to minimize detectability. Even though the poisoned signal and the unpoisoned one are similar, maximizing the MSE greatly affects the distribution of the residuals. Thereby,  it may be more beneficial for the adversary  to directly maximize $\|\begin{bmatrix}\Delta \tilde A_\LS&\Delta \tilde B_\LS\end{bmatrix}\|_2$ (which is a non-convex problem),  while constraining the residuals of the models, to decrease detectability.
A hint comes from the fact that the noise term at time $t$ is $\tilde w_t=w_t+\Delta w_t =w_t + \Delta x_{t+1} - A_\star \Delta x_t - B_\star \Delta u_t$. Since the noise $\tilde w_t$ depends on $\Delta x_{t+1}$, the victim can expect to observe a large value in the correlation of the residuals at lag $1$. This last observation suggests that  an adversary may consider constraining the correlation of the residuals to reduce the detectability.

\paragraph{Correlation tests.}  Consider white process noise  $w_t\sim\mathcal{N}(0,\Sigma_w)$, and let  $\tilde C_\tau = \frac{1}{T}\sum_{t=0}^{T-\tau-1} \tilde R_t\tilde R_{t+\tau}^\top$ be the sample correlation of the residuals at lag $\tau$. 
Under the null hypothesis that the data has not been poisoned, asymptotically we obtain $\sqrt{T} \vect C_\tau \sim \mathcal{N}(0, \Sigma_w \otimes \Sigma_w)$ (from an application of \cite[Lemma 9.A1]{ljung1998system}), from which we  derive the statistics $T \|C_\tau C_0^{-1}\|_{\textnormal{F}}^2 \sim \chi^2(n^2)$. Similarly, following a similar approach as in \cite{hosking1980multivariate}, it is possible to derive the asymptotic Portmanteau statistics to test the whiteness of the residuals $T\sum_{\tau=1}^T\|C_\tau C_0^{-1}\|_{\textnormal{F}}^2 \sim \chi^2(n^2(T-n-m))$. Using these statistics, it is possible to formulate a stealthy attack, as explained in the next section.

\subsection{Stealthy untargeted attack}
\begin{figure}[t]
    \centering
    \includegraphics[width=\linewidth]{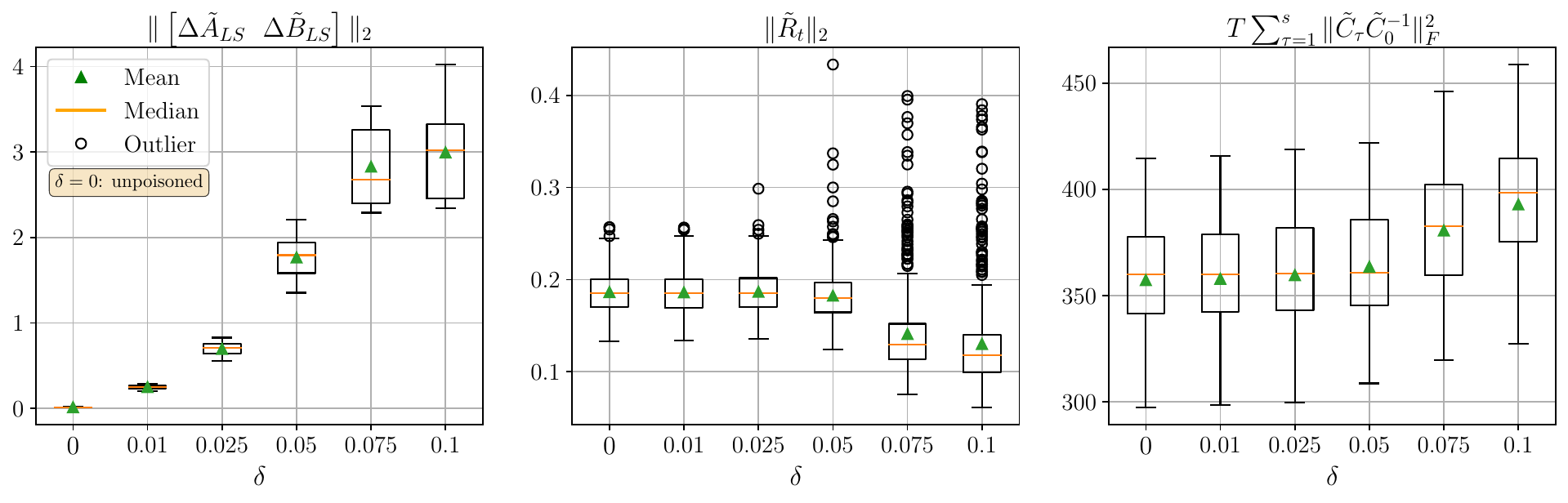}
    \caption{Stealthy attacks on the system used in \cref{example:residuals_variance}. From left to right: distribution over 10 different seeds of the impact on the LS error, $\|\tilde R_t\|_2$ and $T\sum_{\tau=1}^s \|\tilde C_\tau\tilde C_0^{-1}\|_\frobenius^2$.}
    \label{fig:results_stealthy_attack}
\end{figure}
On the basis of the previous findings, we argue that the main quantities of interest to make a poisoning attack stealthy are (1) the norm of the poisoning signals; (2) $\|\tilde R\|_{\textnormal{F}}^2$ the norm of the residuals; (3) $\|\tilde C_\tau \tilde C_0^{-1}\|_{\textnormal{F}}^2,\tau=1,\dots,T-1,$ the norm of the self-normalized correlation terms. Consequently, we propose the following optimization problem to compute poisoning stealthy untargeted attacks:
\begin{equation}\label{eq:optimal_attack}
        \begin{aligned}
            \max_{\Delta U_{-}, \Delta X} \quad & \norm{\begin{bmatrix}\Delta \tilde A_\LS & \Delta  \tilde B_\LS \end{bmatrix}}_2 
            \textrm{ s.t. }  g_i(\data, \Delta U, \Delta X)\leq \delta_i,\quad i=0,2,\dots,s+3,
        \end{aligned}
\end{equation}
with $g_0 = \|\Delta X\|_\textnormal{F}/\|X\|_\textnormal{F}$, $g_1= \|\Delta U\|_\textnormal{F}/\|U\|_\textnormal{F}$, $g_2=\left|1-\|\tilde R\|_{\textnormal{F}}^2/\| R\|_{\textnormal{F}}^2 \right| $, $g_3 = |1-Z_\poisoneddata /Z_\data|$ and $g_{3+\tau} =\left | 1- \|\tilde C_\tau \tilde C_0^{-1}\|_{\textnormal{F}}^2 /\| C_\tau  C_0^{-1}\|_{\textnormal{F}}^2 \right |,\tau=1,\dots, s$, for some $s<T-1$. Intuitively, the constraints limit the relative change of each quantity.
\paragraph{Numerical results.} We applied the attack resulting from (\ref{eq:optimal_attack}) on the system used in \cref{example:residuals_variance}, with $T=500$, $s=25$  and the same value of $\delta$ for all constraints. Since $n=4,m=1$,  there are $2504$ variables to optimize. The optimization problem  is  non-convex, therefore a local solution can be found by means of first-order  methods. Results (see \cref{fig:results_stealthy_attack} and other figures in the appendix) indicate that the resulting poisoning signals can relevantly impact the LS estimate,  while  the statistical indicators (central and right plot in \cref{fig:results_stealthy_attack}) show no evidence of  anomaly. Furthermore, the attack seems to impact have a greater impact on the estimate of $A_\LS$ than that of $B_\LS$ (see also the appendix; for $\delta=0.05$ we obtain $\mathbb{E}\|\Delta \tilde A_\LS\|_2\approx 2$ and $\mathbb{E}\|\Delta \tilde B_\LS\|_2\approx 0.15$). Preliminary results indicate that this effect may be due to the presence of the constraint $g_3$ on $(\mathcal{Z}_\poisoneddata, \mathcal{Z}_\data)$.
Lastly, we observe that for small values of $\delta$ the residuals  do not visually change in a sensible way (refer to the appendix), and an  analysis of the outliers, based on the concept of leverage \cite{kannan2015outlier}, shows  no statistical difference for any values of $\delta$.  These findings suggest that it is possible to devise potentially undetectable poisoning attacks without making use of any sparsity assumption.

\section{Conclusion}
In this work, we have  analyzed poisoning attacks on the data collected from a linear dynamical system affected by process noise. We have focused on the problem of poisoning the least-squares estimate of the underlying dynamical system, which is a quantity used by various data-driven controllers and thus can greatly affect their performance. We have established under which conditions the set  of models compatible with the data includes the true model parameter, and we analyzed the effect of poisoning on the least-squares error.  Based on the analysis of various attack strategies, we have proposed a new stealthy poisoning attack. Results indicate that this attack can relevantly impact the least-squares estimate while being stealthy from a statistical perspective. We conclude that it is possible to craft stealthy attacks that are not necessarily sparse. Possible detection methods include  watermarking  and/or encryption of the data. Future venues of research include analysis of online poisoning attacks; impact and detection of offline poisoning attacks; recovery of the original system matrices based on the set of compatible models.

\section*{Acknowledgments}
This work was supported by the Swedish Foundation for Strategic Research through the CLAS project (grant RIT17-0046).  In addition, the author is grateful to Prof. Alexandre Proutiere for his unwavering support and for providing the opportunity to work on this project.
\ifdefined \arxiv \newpage \else\fi
\bibliography{refs}

\begin{thebibliography}{40}
\providecommand{\natexlab}[1]{#1}
\providecommand{\url}[1]{\texttt{#1}}
\expandafter\ifx\csname urlstyle\endcsname\relax
  \providecommand{\doi}[1]{doi: #1}\else
  \providecommand{\doi}{doi: \begingroup \urlstyle{rm}\Url}\fi

\bibitem[Abbasi-Yadkori and Szepesv{\'a}ri(2011)]{abbasi2011regret}
Yasin Abbasi-Yadkori and Csaba Szepesv{\'a}ri.
\newblock Regret bounds for the adaptive control of linear quadratic systems.
\newblock In \emph{Proceedings of the 24th Annual Conference on Learning
  Theory}, pages 1--26. JMLR Workshop and Conference Proceedings, 2011.

\bibitem[Alfeld et~al.(2016)Alfeld, Zhu, and Barford]{alfeld2016data}
Scott Alfeld, Xiaojin Zhu, and Paul Barford.
\newblock Data poisoning attacks against autoregressive models.
\newblock In \emph{Proceedings of the AAAI Conference on Artificial
  Intelligence}, volume~30, 2016.

\bibitem[Barreno et~al.(2006)Barreno, Nelson, Sears, Joseph, and
  Tygar]{barreno2006can}
Marco Barreno, Blaine Nelson, Russell Sears, Anthony~D Joseph, and J~Doug
  Tygar.
\newblock Can machine learning be secure?
\newblock In \emph{Proceedings of the 2006 ACM Symposium on Information,
  computer and communications security}, pages 16--25, 2006.

\bibitem[Barreno et~al.(2010)Barreno, Nelson, Joseph, and
  Tygar]{barreno2010security}
Marco Barreno, Blaine Nelson, Anthony~D Joseph, and J~Doug Tygar.
\newblock The security of machine learning.
\newblock \emph{Machine Learning}, 81\penalty0 (2):\penalty0 121--148, 2010.

\bibitem[Bhatia et~al.(2015)Bhatia, Jain, and Kar]{bhatia2015robust}
Kush Bhatia, Prateek Jain, and Purushottam Kar.
\newblock Robust regression via hard thresholding.
\newblock \emph{Advances in neural information processing systems}, 28, 2015.

\bibitem[Bhatia et~al.(2017)Bhatia, Jain, Kamalaruban, and
  Kar]{bhatia2017consistent}
Kush Bhatia, Prateek Jain, Parameswaran Kamalaruban, and Purushottam Kar.
\newblock Consistent robust regression.
\newblock \emph{Advances in Neural Information Processing Systems}, 30, 2017.

\bibitem[Biggio et~al.(2012)Biggio, Nelson, and Laskov]{biggio2012poisoning}
Battista Biggio, Blaine Nelson, and Pavel Laskov.
\newblock Poisoning attacks against support vector machines.
\newblock In \emph{Proceedings of the 29th International Coference on
  International Conference on Machine Learning}, pages 1467--1474, 2012.

\bibitem[Bittanti et~al.(1992)Bittanti, Campi, and
  Lorito]{bittanti1992effective}
Sergio Bittanti, Marco Campi, and Fabrizio Lorito.
\newblock Effective identification algorithms for adaptive control.
\newblock \emph{International journal of adaptive control and signal
  processing}, 6\penalty0 (3):\penalty0 221--235, 1992.

\bibitem[Box et~al.(2015)Box, Jenkins, Reinsel, and Ljung]{box2015time}
George~EP Box, Gwilym~M Jenkins, Gregory~C Reinsel, and Greta~M Ljung.
\newblock \emph{Time series analysis: forecasting and control}.
\newblock John Wiley \& Sons, 2015.

\bibitem[Campi et~al.(2002)Campi, Lecchini, and Savaresi]{campi2002virtual}
Marco~C Campi, Andrea Lecchini, and Sergio~M Savaresi.
\newblock Virtual reference feedback tuning: a direct method for the design of
  feedback controllers.
\newblock \emph{Automatica}, 38\penalty0 (8):\penalty0 1337--1346, 2002.

\bibitem[Chekan and Langbort(2020)]{chekan2020regret}
Jafar~Abbaszadeh Chekan and Cedric Langbort.
\newblock Regret bounds for lq adaptive control under database attacks
  (extended version).
\newblock \emph{arXiv preprint arXiv:2004.00241}, 2020.

\bibitem[Coulson et~al.(2019)Coulson, Lygeros, and
  D{\"o}rfler]{coulson2019data}
Jeremy Coulson, John Lygeros, and Florian D{\"o}rfler.
\newblock Data-enabled predictive control: In the shallows of the deepc.
\newblock In \emph{2019 18th European Control Conference (ECC)}, pages
  307--312. IEEE, 2019.

\bibitem[De~Boer et~al.(2005)De~Boer, Kroese, Mannor, and
  Rubinstein]{de2005tutorial}
Pieter-Tjerk De~Boer, Dirk~P Kroese, Shie Mannor, and Reuven~Y Rubinstein.
\newblock A tutorial on the cross-entropy method.
\newblock \emph{Annals of operations research}, 134\penalty0 (1):\penalty0
  19--67, 2005.

\bibitem[De~Persis and Tesi(2019)]{de2019formulas}
Claudio De~Persis and Pietro Tesi.
\newblock Formulas for data-driven control: Stabilization, optimality, and
  robustness.
\newblock \emph{IEEE Transactions on Automatic Control}, 65\penalty0
  (3):\penalty0 909--924, 2019.

\bibitem[Drouiche(2000)]{drouiche2000new}
K~Drouiche.
\newblock A new test for whiteness.
\newblock \emph{IEEE transactions on signal processing}, 48\penalty0
  (7):\penalty0 1864--1871, 2000.

\bibitem[Feng and Lavaei(2021)]{feng2021learning}
Han Feng and Javad Lavaei.
\newblock Learning of dynamical systems under adversarial attacks.
\newblock In \emph{2021 60th IEEE Conference on Decision and Control (CDC)},
  pages 3010--3017. IEEE, 2021.

\bibitem[Hjalmarsson(2002)]{hjalmarsson2002iterative}
H{\aa}kan Hjalmarsson.
\newblock Iterative feedback tuning—an overview.
\newblock \emph{International journal of adaptive control and signal
  processing}, 16\penalty0 (5):\penalty0 373--395, 2002.

\bibitem[Hosking(1980)]{hosking1980multivariate}
Jonathan~RM Hosking.
\newblock The multivariate portmanteau statistic.
\newblock \emph{Journal of the American Statistical Association}, 75\penalty0
  (371):\penalty0 602--608, 1980.

\bibitem[Jagielski et~al.(2018)Jagielski, Oprea, Biggio, Liu, Nita-Rotaru, and
  Li]{jagielski2018manipulating}
Matthew Jagielski, Alina Oprea, Battista Biggio, Chang Liu, Cristina
  Nita-Rotaru, and Bo~Li.
\newblock Manipulating machine learning: Poisoning attacks and countermeasures
  for regression learning.
\newblock In \emph{2018 IEEE Symposium on Security and Privacy (SP)}, pages
  19--35. IEEE, 2018.

\bibitem[Kannan and Manoj(2015)]{kannan2015outlier}
K~Senthamarai Kannan and K~Manoj.
\newblock Outlier detection in multivariate data.
\newblock \emph{Applied Mathematical Sciences}, 47\penalty0 (9):\penalty0
  2317--2324, 2015.

\bibitem[Karimi et~al.(2004)Karimi, Mi{\v{s}}kovi{\'c}, and
  Bonvin]{karimi2004iterative}
A~Karimi, L~Mi{\v{s}}kovi{\'c}, and D~Bonvin.
\newblock Iterative correlation-based controller tuning.
\newblock \emph{International journal of adaptive control and signal
  processing}, 18\penalty0 (8):\penalty0 645--664, 2004.

\bibitem[Kleinbaum et~al.(2013)Kleinbaum, Kupper, Nizam, and
  Rosenberg]{kleinbaum2013applied}
David~G Kleinbaum, Lawrence~L Kupper, Azhar Nizam, and Eli~S Rosenberg.
\newblock \emph{Applied regression analysis and other multivariable methods}.
\newblock Cengage Learning, 2013.

\bibitem[Kloft and Laskov(2010)]{kloft2010online}
Marius Kloft and Pavel Laskov.
\newblock Online anomaly detection under adversarial impact.
\newblock In \emph{Proceedings of the thirteenth international conference on
  artificial intelligence and statistics}, pages 405--412. JMLR Workshop and
  Conference Proceedings, 2010.

\bibitem[Koch et~al.(2020)Koch, Berberich, and Allg{\"o}wer]{koch2020verifying}
Anne Koch, Julian Berberich, and Frank Allg{\"o}wer.
\newblock Verifying dissipativity properties from noise-corrupted input-state
  data.
\newblock In \emph{2020 59th IEEE Conference on Decision and Control (CDC)},
  pages 616--621. IEEE, 2020.

\bibitem[Liu et~al.(2017)Liu, Ma, Aafer, Lee, Zhai, Wang, and
  Zhang]{liu2017trojaning}
Yingqi Liu, Shiqing Ma, Yousra Aafer, Wen-Chuan Lee, Juan Zhai, Weihang Wang,
  and Xiangyu Zhang.
\newblock Trojaning attack on neural networks.
\newblock 2017.

\bibitem[Ljung(1998)]{ljung1998system}
Lennart Ljung.
\newblock System identification.
\newblock In \emph{Signal analysis and prediction}, pages 163--173. Springer,
  1998.

\bibitem[Massey~Jr(1951)]{massey1951kolmogorov}
Frank~J Massey~Jr.
\newblock The kolmogorov-smirnov test for goodness of fit.
\newblock \emph{Journal of the American statistical Association}, 46\penalty0
  (253):\penalty0 68--78, 1951.

\bibitem[Nelson(1998)]{nelson1998anderson}
Lloyd~S Nelson.
\newblock The anderson-darling test for normality.
\newblock \emph{Journal of Quality Technology}, 30\penalty0 (3):\penalty0 298,
  1998.

\bibitem[Nguyen and Tran(2013)]{nguyen2013exact}
Nam~H Nguyen and Trac~D Tran.
\newblock Exact recoverability from dense corrupted observations via
  $\ell_1$-minimization.
\newblock \emph{IEEE transactions on information theory}, 59\penalty0
  (4):\penalty0 2017--2035, 2013.

\bibitem[Rubinstein et~al.(2009)Rubinstein, Nelson, Huang, Joseph, Lau, Rao,
  Taft, and Tygar]{rubinstein2009antidote}
Benjamin~IP Rubinstein, Blaine Nelson, Ling Huang, Anthony~D Joseph, Shing-hon
  Lau, Satish Rao, Nina Taft, and J~Doug Tygar.
\newblock Antidote: understanding and defending against poisoning of anomaly
  detectors.
\newblock In \emph{Proceedings of the 9th ACM SIGCOMM Conference on Internet
  Measurement}, pages 1--14, 2009.

\bibitem[Russo and Proutiere(2021)]{russo2021poisoning}
Alessio Russo and Alexandre Proutiere.
\newblock Poisoning attacks against data-driven control methods.
\newblock In \emph{2021 American Control Conference (ACC)}, pages 3234--3241.
  IEEE, 2021.

\bibitem[Russo et~al.(2021)Russo, Molinari, and Proutiere]{russo2021data}
Alessio Russo, Marco Molinari, and Alexandre Proutiere.
\newblock Data-driven control and data-poisoning attacks in buildings: the kth
  live-in lab case study.
\newblock In \emph{2021 29th Mediterranean Conference on Control and Automation
  (MED)}, pages 53--58. IEEE, 2021.

\bibitem[Shafahi et~al.(2018)Shafahi, Huang, Najibi, Suciu, Studer, Dumitras,
  and Goldstein]{shafahi2018poison}
Ali Shafahi, W~Ronny Huang, Mahyar Najibi, Octavian Suciu, Christoph Studer,
  Tudor Dumitras, and Tom Goldstein.
\newblock Poison frogs! targeted clean-label poisoning attacks on neural
  networks.
\newblock \emph{Advances in neural information processing systems}, 31, 2018.

\bibitem[Shen et~al.(2016)Shen, Diamond, Gu, and Boyd]{shen2016disciplined}
Xinyue Shen, Steven Diamond, Yuantao Gu, and Stephen Boyd.
\newblock Disciplined convex-concave programming.
\newblock In \emph{2016 IEEE 55th Conference on Decision and Control (CDC)},
  pages 1009--1014. IEEE, 2016.

\bibitem[Showkatbakhsh et~al.(2016)Showkatbakhsh, Tabuada, and
  Diggavi]{showkatbakhsh2016secure}
Mehrdad Showkatbakhsh, Paulo Tabuada, and Suhas Diggavi.
\newblock Secure system identification.
\newblock In \emph{2016 54th Annual Allerton Conference on Communication,
  Control, and Computing (Allerton)}, pages 1137--1141. IEEE, 2016.

\bibitem[Tian et~al.(2022)Tian, Cui, Liang, and Yu]{tian2022comprehensive}
Zhiyi Tian, Lei Cui, Jie Liang, and Shui Yu.
\newblock A comprehensive survey on poisoning attacks and countermeasures in
  machine learning.
\newblock \emph{ACM Computing Surveys (CSUR)}, 2022.

\bibitem[Willems et~al.(2005)Willems, Rapisarda, Markovsky, and
  De~Moor]{willems2005note}
Jan~C Willems, Paolo Rapisarda, Ivan Markovsky, and Bart~LM De~Moor.
\newblock A note on persistency of excitation.
\newblock \emph{Systems \& Control Letters}, 54\penalty0 (4):\penalty0
  325--329, 2005.

\bibitem[Xiao et~al.(2015)Xiao, Biggio, Brown, Fumera, Eckert, and
  Roli]{xiao2015feature}
Huang Xiao, Battista Biggio, Gavin Brown, Giorgio Fumera, Claudia Eckert, and
  Fabio Roli.
\newblock Is feature selection secure against training data poisoning?
\newblock In \emph{international conference on machine learning}, pages
  1689--1698. PMLR, 2015.

\bibitem[Yu et~al.(2022)Yu, Zhao, Chinchali, and Topcu]{2209.09108}
Yue Yu, Ruihan Zhao, Sandeep Chinchali, and Ufuk Topcu.
\newblock Poisoning attacks against data-driven predictive control, 2022.
\newblock URL \url{https://arxiv.org/abs/2209.09108}.

\bibitem[Zhang et~al.(2020)Zhang, Zhu, and Lessard]{zhang2020online}
Xuezhou Zhang, Xiaojin Zhu, and Laurent Lessard.
\newblock Online data poisoning attacks.
\newblock In \emph{Learning for Dynamics and Control}, pages 201--210. PMLR,
  2020.

\end{thebibliography}
\ifdefined \arxiv
\newpage
\appendix
\section{Appendix}
\subsection{Proofs}
\begin{proof}[Proof of Lemma \ref{lemma:compatible_a_b}]
The claim of compatibility follows from the text in the main part of this manuscript (alternatively see also  \cite{koch2020verifying}). To prove the latter claim, assume first that $(A_\star, B_\star)\in \Sigma_\poisoneddata$. Then $\tilde W_-$ is  given by
\[
\tilde W_- = \begin{bmatrix}
    I_n & -A_\star & -B_\star
    \end{bmatrix} \begin{bmatrix}
    \tilde X_+ \\ \tilde X_- \\ \tilde U_- 
    \end{bmatrix}= W_- +\begin{bmatrix}
    I_n & -A_\star & -B_\star
    \end{bmatrix} \begin{bmatrix}
    \Delta X_+ \\ \Delta X_- \\ \Delta U_- 
    \end{bmatrix}=\tilde W_\star
\]
 where the latter equality follows from \cref{eq:fundamental_relationship}.  Since $(A_\star, B_\star)\in \Sigma_\poisoneddata$, then $\tilde W_-\in \mathcal{W}_\poisoneddata\subseteq \mathcal{W}^T$.
 
 Consider the reverse direction, and assume that  $\tilde W_\star\in \mathcal{W}^T$. Then $\tilde W_\star
\in \mathcal{W}_\poisoneddata$ since
\begin{align*}
     \tilde W_\star \tilde G^\perp &=\left( W_- +\begin{bmatrix}
    I_n & -A_\star & -B_\star
    \end{bmatrix} \begin{bmatrix}
    \Delta X_+ \\ \Delta X_- \\ \Delta U_- 
    \end{bmatrix}\right)\begin{bmatrix}\tilde X_+\\\tilde X_-\\ \tilde U_-
 \end{bmatrix}^\perp,\\
 &=\left( \begin{bmatrix}
    I_n & -A_\star & -B_\star
    \end{bmatrix} \left(\begin{bmatrix}
    X_+ \\  X_- \\  U_- 
    \end{bmatrix} + \begin{bmatrix}
    \Delta X_+ \\ \Delta X_- \\ \Delta U_- 
    \end{bmatrix}\right)\right)\begin{bmatrix}\tilde X_+\\\tilde X_-\\ \tilde U_-
 \end{bmatrix}^\perp,\\
 &=\begin{bmatrix}
    I_n & -A_\star & -B_\star
    \end{bmatrix} \begin{bmatrix}\tilde X_+\\
    \tilde X_- \\ \tilde U_- 
    \end{bmatrix}\begin{bmatrix}\tilde X_+\\\tilde X_-\\ \tilde U_-
 \end{bmatrix}^\perp \\ &=0.
\end{align*}
From which follows that $(A_\star, B_\star)\in \Sigma_\poisoneddata$.
\end{proof}
\begin{proof}[Proof of lemma \ref{lemma:estimate_delta}]
From lemma \ref{lemma:compatible_a_b} we know that for a poisoned dataset $\poisoneddata=(\tilde U_-, \tilde X)$  the corresponding noise realization is given by $\tilde W_\star$. Therefore, using lemma \ref{lemma:ls_error} we obtain the result. Alternatively, note that the LS estimator $(\hat{\tilde A},\hat{\tilde B})$ of $\poisoneddata$  satisfies :
\[
\begin{bmatrix}
 \hat{\tilde A} & \hat{\tilde B}
\end{bmatrix}\begin{bmatrix}
\tilde X_-\\ \tilde U_-
\end{bmatrix} = \tilde X_+.
\]
Let $\Delta A = \hat{\tilde {A}}-A_\star$, and similarly define $\Delta B$. Then  the previous expression becomes
\[
\begin{bmatrix}
 A_\star+\Delta A & B_\star+\Delta B
\end{bmatrix} \begin{bmatrix}
 X_-+\Delta X_-\\  U_-+\Delta U_-
\end{bmatrix} = X_+  +\Delta X_+,
\]
from which we obtain
$
\begin{bmatrix}
 A_\star & B_\star
\end{bmatrix} \begin{bmatrix}
\Delta X_-\\ \Delta U_-
\end{bmatrix}+
\begin{bmatrix}
\Delta A &\Delta B
\end{bmatrix} \begin{bmatrix}
\tilde X_-\\ \tilde U_-
\end{bmatrix} = W_- +\Delta X_+.$
The result follows by using the expression of $\tilde W_\star$, thus
\[
\begin{bmatrix}
\Delta A &\Delta B
\end{bmatrix} \begin{bmatrix}
\tilde X_-\\ \tilde U_-
\end{bmatrix} = W_- +\tilde W_\star.
\]
The second part follows  from the fact that $\norm{\Delta \tilde \theta_\LS}_2 = \norm{\vect (\Delta \tilde A_\LS, \Delta \tilde B_\LS)}_2 = \norm{(\Delta \tilde A_\LS, \Delta \tilde B_\LS)}_{\textnormal{F}} $, thus
    \begin{align*}
        \|\Delta \tilde \theta_\LS\|_2^2 &= \|(\Delta \tilde A_\LS, \Delta \tilde B_\LS)\|_\frobenius^2,\\
        &= \trace\left((\tilde \Psi_-\tilde \Psi_-^\top)^{-1} \tilde \Psi_- \tilde W_\star^\top \tilde W_\star \tilde \Psi_-^\top (\tilde \Psi_-\tilde \Psi_-^\top)^{-1} \right),\\
        &= \trace\left( \tilde W_\star^\top \tilde W_\star \tilde \Psi_-^\top (\tilde \Psi_-\tilde \Psi_-^\top)^{-1}(\tilde \Psi_-\tilde \Psi_-^\top)^{-1} \tilde \Psi_- \right),\\
        &\stackrel{(a)}{\leq} \sigma_{\text{max}}(\tilde W_\star)^2 \trace\left(\tilde \Psi_-^\top (\tilde \Psi_-\tilde \Psi_-^\top)^{-1}(\tilde \Psi_-\tilde \Psi_-^\top)^{-1} \tilde \Psi_-  \right),\\
        &= \sigma_{\text{max}}(\tilde W_\star)^2 \trace\left((\tilde \Psi_-\tilde \Psi_-^\top)^{-1} \right),\\
        &\stackrel{(b)}{\leq} (\sigma_{\text{max}}(W_-) + \sigma_{\text{max}}(\Delta W_-))^2 / \left ( \sum_{t=0}^{t-1} \|\tilde \psi_t\|_2^2 \right).
    \end{align*}
    In the above proof (a) follows from the Von Neumann's trace inequality $\trace(XY) \leq \sum_{i} \sigma_i(X)\sigma_i(Y) \leq \|Y\|_2\trace(X)$, where $Y=W_\star^\top \tilde W_\star$ and $X=\tilde \Psi_-^\top (\tilde \Psi_-\tilde \Psi_-^\top)^{-1}(\tilde \Psi_-\tilde \Psi_-^\top)^{-1} \tilde \Psi_-$; (b) follows from $\|\tilde W_\star\|_2 = \| W_-+\Delta W_-\|_2 \leq \| W_-\|_2 +\| \Delta W_-\|_2 $ and the fact that $\tilde \Psi_-\tilde \Psi_-^\top$ it's a symmetric positive-definite matrix, thus diagonalizable as $\tilde \Psi_-\tilde \Psi_-^\top = VDV^\top$, and consequently $\trace((\tilde \Psi_-\tilde \Psi_-^\top)^{-1}) = \trace(D^{-1})= \trace(D)^{-1}$.
    The reverse direction follows by noting that, similarly, $\|XY\|_\frobenius \geq \sigma_{\text{min}}(Y) \|X\|_\frobenius .$
\end{proof}
\begin{proof}[Proof of corollary \ref{corollary:lower_bound_error_exploration}.]
Let $\{v_i\}_{i=1}^{n(n+m)}$ be an orthonormal basis of $\mathbb{R}^{n(n+m)}$ with $v_1 = \theta_\LS/ \|\theta_\LS\|_2$ and let $V=\begin{bmatrix}v_1 &v_2&\dots &v_{n(n+m)}\end{bmatrix}$. Then, we verify that
\[
(\Psi_-\Psi_-^\top \otimes I_n)=\sum_{t=0}^{T-1}(\tilde\psi_t\tilde\psi_t^\top \otimes I_n) = \sum_{t=0}^{T-1}(\tilde\psi_t \otimes I_n) (\tilde\psi_t\otimes I_n)^\top = VDV^\top,
\]
for some $D$ where $(D)_{kk}=\sum_{t=0}^{T-1}v_k^\top  (\tilde\psi_t \otimes I_n) (\tilde\psi_t\otimes I_n)^\top v_k$. Using that $V_k =\vect^{-1}(v_k)$, where $V_k$ is an $n\times n+m$ matrix, we find
\[
(\tilde\psi_t\otimes I_n)^\top v_k = (\tilde\psi_t^\top \otimes I_n)v_k = V_k \tilde \psi_t\Rightarrow (D)_{kk} = \sum_{t=0}^{T-1} \|V_k \tilde \psi_t\|_2^2 = \|V_k \tilde \Psi_-\|_{\textnormal{F}}^2.
\]
Observing that $\Delta\theta_\LS = \vect\left(\tilde W_\star \Psi_-^\top (\Psi_-\Psi_-^\top)^{-1}\right) = (\Psi_-\Psi_-^\top \otimes I_n)^{-1}\vect(\tilde W_\star \Psi_-^\top)$, and that $\vect(\tilde W_\star \Psi_-^\top) = (\Psi_-\otimes I_n) \vect(\tilde W_\star)$, we conclude
\begin{align*}
    v_k^\top \Delta \theta_\LS &= v_k^\top(\Psi_-\Psi_-^\top \otimes I_n)^{-1}\vect(\tilde W_\star \Psi_-^\top),\\
    &= v_k^\top VD^{-1}V^\top \vect(\tilde W_\star \Psi_-^\top),\\
    &= (D)_{kk}^{-1} v_k^\top \vect(\tilde W_\star \Psi_-^\top),\\
    &= (D)_{kk}^{-1} v_k^\top (\Psi_-\otimes I_n) \vect(\tilde W_\star),\\ 
    &= (D)_{kk}^{-1} \vect(V_k \Psi_-)^\top  \vect(\tilde W_\star)
\end{align*}
Letting $\alpha$ be the angle between $\vect(V_k \Psi_-)$ and $\vect(\tilde W_\star)$ we derive
\[
|\vect(V_k \Psi_-)^\top  \vect(\tilde W_\star)|^2 = \|V_k \Psi_-\|_{\textnormal{F}}^2 \|\tilde W_\star\|_{\textnormal{F}}^2 \cos(\alpha)^2 \geq \|V_k \Psi_-\|_{\textnormal{F}}^2  \cos(\alpha)^2 (\sigma_{\text{min}}(W_-) +\sigma_{\text{min}}(\Delta W_-))^2
\]
Finally, we conclude that
\[
\sqrt{( v_k^\top \Delta \theta_\LS)^2} \geq \frac{|\cos(\alpha)| (\sigma_{\text{min}}(W_-) +\sigma_{\text{min}}(\Delta W_-))}{\|V_k \Psi_-\|_{\textnormal{F}}}.
\]

\end{proof}
\begin{proof}[Proof of lemma \ref{lemma:input_attack_indistinguishable_asymptotics}]
Define the covariance matrix
\[S(T) = \frac{1}{T}\begin{bmatrix}
 X_- \\ \tilde U_-
\end{bmatrix}\begin{bmatrix}
 X_-^\top & \tilde U_-^\top
\end{bmatrix}= \frac{1}{T}\left(\sum_{t=0}^{T-1}\begin{bmatrix}
 x_t x_t^\top &  x_t  a_t^\top \\
 a_t  x_t^\top &  a_t  a_t^\top
\end{bmatrix}\right).\]
Since $(a_t, x_t)$ are independent,  it follows that $S(T)$ for $T\to\infty$ converges to a block-diagonal matrix 
\[\lim_{T\to\infty} S(T)  = S=\begin{bmatrix}
\Sigma_x & 0\\ 0 &\Sigma_u
\end{bmatrix},
\]
whose  elements are the covariance matrix $\Sigma_x$  of the state $x_t$ (at stationarity),  and the covariance $\Sigma_u$ of $a_t$. Now, from lemma \ref{lemma:estimate_delta} observe that
\[
\begin{bmatrix}
 \Delta A & \Delta B
\end{bmatrix}S(T) = \frac{1}{T}\sum_{t=0}^{T-1} \left(w_t  - B_\star (-u_t+a_t) \right)\begin{bmatrix}
 x_t^\top &  a_t^\top
\end{bmatrix}
\]
Then, we see that
$
\frac{1}{T}\sum_{t=0}^{T-1} (w_t- B_\star (-u_t +a_t)) a_t^\top \to -B_\star \Sigma_u.
$
Consequently, as $T\to\infty$, the error $\Delta B$ converges to $-B_\star
$ w.p. $1$.
\end{proof}
\begin{proof}[Proof of lemma \ref{lemma:MSE_residuals_under_attack}]
Using lemma \ref{lemma:residuals} we have that $\tilde R = \tilde W_\star F_T$, where $F_T=(I_T - \tilde M)$ and $\tilde M=\begin{bmatrix}
    \tilde X_- \\ \tilde U_- 
    \end{bmatrix}^\dagger \begin{bmatrix}
    \tilde X_- \\ \tilde U_- 
    \end{bmatrix}$ is a $T\times T$ idempotent matrix. From this result it also follows that in absence of an attack $\vect R = \vect( W_- (I-M)) = ((I-M)^\top \otimes I_n) \vect W_-$, and thus $\vect R \sim \mathcal{N}(0,  ((I-M)^\top \otimes I_n)(\Sigma_w \otimes I_T) ((I-M)^\top \otimes I_n)^\top)$.

    From the rank condition on $\poisoneddata$, we find that  $I_T-M$ has $T-(n+m)$ eigenvalues that are $1$, and $n+m$ that are 0. 
    Consequently, since $F_TF_T^\top = F_T$ we find
    \[
        \|\tilde R\|_{\textnormal{F}}^2 = \trace(F_T^\top \tilde W_\star^\top \tilde W_\star F_T ) =\trace(\tilde W_\star F_T \tilde W_\star^\top ).
    \]
    Then $\|\tilde R\|_{\textnormal{F}}^2 = \trace(W_- F_TW_-^\top) +\trace(\Delta W_- F_T\Delta W_-^\top) +2\trace( W_- F_T \Delta W_-^\top)$. From the fact that $F_T$ is symmetric, it can be diagonalized by an orthogonal matrix $U_T$ s.t. $F_T =U_T^\top D_T U_T$, where $D_T$ has $T-(n+m)$ ones and $n+m$ zeros along the diagonal. Consequently $\trace(W_- F_TW_-^\top) = \sum_{t=0}^{T-n-m-1} \|\bar w_t\|_2^2$, where $\bar w_t$ are the components of $W_-U_T^\top$, which are also independent and normally distributed according to $\mathcal{N}(0,\Sigma_w)$. Therefore $W_- F_TW_-^\top $ is distributed according to a  Wishart distribution with $T-n-m$ degrees of freedom, and scale matrix $\Sigma_w$. Let $\lambda_i$ be the $i$-th eigenvalue of $\Sigma$, then $\trace(W_- F_TW_-^\top) \sim \sum_{i=1}^n \lambda_i \chi^2(T-n-m)$.  
    We also derive that  $\mathbb{E}[\trace(W_- F_TW_-^\top)] = (T-n-m)\trace(\Sigma_w)$.  
    
    Continuing the proof,  from the assumption of independence we find  $\mathbb{E}[\trace( W_- F_T \Delta W_-^\top)] =0$. Finally, observe the following upper bound $\trace(\Delta W_- F_T\Delta W_-^\top) = \|\Delta W_- F_T\|_{\textnormal{F}}^2 \leq \|\Delta W_-\|_{\textnormal{F}}^2 = \sum_{t=0}^{T-1}\| \Delta x_{t+1}- A_\star \Delta x_t - B_\star \Delta u_t\|_2^2$. Taking the expectation of this  term concludes the proof.
\end{proof}
\begin{proof}[Proof of lemma \ref{lemma:attack_sensitivity}]
From the orthogonality of the LS estimator we have that
\[
\tilde R\tilde\Psi_-^\top =0 \hbox{ and }  R\Psi_-^\top=0
\]
therefore $\tilde R\tilde\Psi_-^\top =0 - R\Psi_-^\top=0\Rightarrow (\tilde R - R)\tilde \Psi_-^\top = R\Delta\Psi_-^\top$. Hence, we get the following set of inequalities (which follow from a similar bound provided in lemma \ref{lemma:estimate_delta}):
\begin{align*}
\|(\tilde R - R)\|_\textnormal{F}^2 &= \|R \Delta\Psi_-^\top (\tilde \Psi_-^\top)^\dagger\|_\textnormal{F}^2,\\
&\leq \|R\|_{\textnormal{F}}^2 \| \Delta\Psi_-^\top (\tilde \Psi_-^\top)^\dagger\|_\textnormal{F}^2,\\
&\leq \|R\|_{\textnormal{F}}^2 \sigma_{\text{max}}(\Delta \Psi_-)^2 \| (\tilde \Psi_-^\top)^\dagger\|_\textnormal{F}^2,\\
&\leq \|R\|_{\textnormal{F}}^2 \frac{\sigma_{\text{max}}(\Delta \Psi_-)^2}{ \|\tilde \Psi_-\|_\textnormal{F}^2}.
\end{align*}
\end{proof}

\begin{lemma}\label{lemma:ls_error}
The error made by the LS estimate given a dataset $\data=(U_-,X)$ is \begin{equation}
\begin{bmatrix}
 \hat A - A_\star &\hat B- B_\star
\end{bmatrix}= W_- \begin{bmatrix}
 X_-\\ U_-
\end{bmatrix}^\dagger.\end{equation} Furthermore, the LS estimate is compatible with the data, i.e. $(\hat A,\hat B)\in \Sigma_\data$, with $\tilde W_-=0$.
\end{lemma}
\begin{proof}
The proof of the first claim follows from the following sequence of equalities:  
\begin{align*}
\begin{bmatrix}
 \hat A  &\hat B
\end{bmatrix} = X_+ \begin{bmatrix}
X_-\\ U_-
\end{bmatrix}^\dagger = \left( W_- + \begin{bmatrix}A_\star & B_\star\end{bmatrix}\begin{bmatrix}X_-\\ U_-\end{bmatrix}\right) \begin{bmatrix}
X_-\\ U_-
\end{bmatrix}^\dagger = W_-\begin{bmatrix}
X_-\\ U_-
\end{bmatrix}^\dagger + \begin{bmatrix}
 A_\star &B_\star
\end{bmatrix}.
\end{align*}
To verify if $(\hat A,\hat B)\in \Sigma_\data$, note that the following condition needs to hold $X_+=\hat A X_-+ \hat B U_- + \tilde W_-$ for some $\tilde W_-\in \mathcal{W}_\data$. In light of the previous result, the condition can be rewritten as $A_\star X_- + B_\star U_- +W_- = \hat A X_-+ \hat B U_- +\tilde W_- \Rightarrow W_-= W_-+\tilde W_-$. Therefore, it must be $\tilde W_-=0$. Since $0\in \mathcal{W}_\data$ we conclude that the LS-estimate belongs to $\Sigma_\data$.
\end{proof}

\begin{lemma}\label{lemma:residuals}
Consider a poisoned dataset $\poisoneddata$ and its LS estimate  $(\tilde A_\LS, \tilde B_\LS)$ . Then the residuals $\tilde R = \begin{bmatrix}
 \tilde R_0 &\tilde R_1 &\dots &\tilde R_{T-1}
\end{bmatrix}=\tilde X_+ - \begin{bmatrix}\tilde A_\LS & \tilde B_\LS\end{bmatrix}\begin{bmatrix}
    \tilde X_- \\ \tilde U_- 
    \end{bmatrix}$ satisfy 
    \begin{equation}
        \tilde R = \tilde W_\star (I-\tilde M),
    \end{equation}
    where $\tilde M=\begin{bmatrix}
    \tilde X_- \\ \tilde U_- 
    \end{bmatrix}^\dagger \begin{bmatrix}
    \tilde X_- \\ \tilde U_- 
    \end{bmatrix}$ is a $T\times T$ idempotent matrix.
\end{lemma}
\begin{proof}
The proof relies on using lemma \ref{lemma:estimate_delta} and the expression of $\tilde W_\star$:
\begin{align*}
     \tilde R &= \tilde X_+ - \begin{bmatrix}\tilde A_\LS & \tilde B_\LS\end{bmatrix}\begin{bmatrix}
    \tilde X_- \\ \tilde U_- 
    \end{bmatrix},\\
    &\stackrel{(a)}{=} \tilde X_+ - \begin{bmatrix}A_\star & B_\star\end{bmatrix}\begin{bmatrix}
    \tilde X_- \\ \tilde U_- 
    \end{bmatrix} - \tilde W_\star \tilde M,\\
    &= X_+ +\Delta X_+- \begin{bmatrix}A_\star & B_\star\end{bmatrix}\begin{bmatrix}
     X_- \\  U_- 
    \end{bmatrix}-\begin{bmatrix}A_\star & B_\star\end{bmatrix}\begin{bmatrix}
     \Delta X_- \\  \Delta U_- 
    \end{bmatrix} - \tilde W_\star \tilde M,\\
    &\stackrel{(b)}{=} X_+ - \begin{bmatrix}A_\star & B_\star\end{bmatrix}\begin{bmatrix}
     X_- \\  U_- 
    \end{bmatrix}  - W_- + \tilde W_\star (I-\tilde M),\\
    &\stackrel{(c)}{=}W_-  - W_- + \tilde W_\star (I-\tilde M) = \tilde W_\star (I-\tilde M).
\end{align*}
(a) is an application of \ref{lemma:estimate_delta}; (b) follows from the fact that $\tilde W_\star -W_- = \Delta W_- = \Delta X_+ - A_\star \Delta X_--B_\star \Delta U_-$; (c) is a consequence of \cref{eq:fundamental_relationship}. To see that $M$ is an idempotent matrix it is sufficient to note that $\tilde M = \begin{bmatrix} \tilde X_-^\top & \tilde U_-^\top\end{bmatrix}\begin{bmatrix}
    \tilde X_- \tilde X_-^\top  & \tilde X_- \tilde U_-^\top \\ \tilde U_- \tilde X_-^\top & \tilde U_-\tilde U_-^\top
    \end{bmatrix}^{-1} \begin{bmatrix}
    \tilde X_- \\ \tilde U_- 
    \end{bmatrix}$.
\end{proof}
\newpage
\subsection{Numerical results}
In this section, we illustrate the details of the numerical results presented in the paper.
Please, find all the code at the following link: \url{https://github.com/rssalessio/data-poisoning-linear-systems}.
\subsubsection{Input poisoning attack}

In \cref{example:input_poisoning} we explored the effects of input poisoning. The choice of a scalar system allows to visualize the confidence regions of the parameters for different number of samples. As a reminder, the true system is described by the equation
\begin{equation}
    x_{t+1}=\underbrace{0.7}_{a} x_t+\underbrace{0.5}_{b} u_t+w_t,\quad w_t\sim \mathcal{N}(0,\sigma^2).
\end{equation}
In the unpoisoned case, the LS estimate is distributed according to a Gaussian distribution, of covariance $\sigma^2( \Psi_-  \Psi_-^\top)^{-1} $. When the true variance is unknown, an estimate $\hat \sigma_T$ can be computed from the MSE using $T$ samples. Then, for $\theta_\LS=(a_\LS, b_\LS)$, the quantity
\begin{equation}
    \frac{1}{\hat \sigma_T} (\theta_\LS - \theta_\star)^\top (\Psi_-\Psi_-^\top)(\theta_\LS - \theta_\star) \in F(2,T-2)
\end{equation}
is distributed according to a $F(2,T-2)$ distribution \cite{ljung1998system}, which can be used to derive a confidence region for $(a,b)$.
In case the data is poisoned, the estimate $\hat \sigma_T$ will be different from the true value $\sigma$. In fact, for this type of attack where the true signal $u_t$ is replaced by another signal with the same statistical properties, the variance will be higher, as explained in section \ref{sec:residual_analysis}. In conclusion, we expect to obtain bigger confidence regions in the poisoned case, which is indeed the result that we obtain from numerical results (see also the figures in \ref{fig:input_attack_appendix}).

\begin{table}[h]
\centering
\begin{tabular}{l||cc|cc|cc} 
\toprule
\multicolumn{1}{c||}{\multirow{2}{*}{$\sigma/T$}} & \multicolumn{2}{c|}{$T=30$}                                                                              & \multicolumn{2}{c|}{$T=100$}                                                                             & \multicolumn{2}{c}{$T=1000$}                                                                             \\ 
\cline{2-7}
\multicolumn{1}{c||}{}                            & \multicolumn{1}{c}{$\mathcal{Z}_{\mathcal{D}}$} & \multicolumn{1}{c|}{$\mathcal{Z}_{\poisoneddata}$} & \multicolumn{1}{c}{$\mathcal{Z}_{\mathcal{D}}$} & \multicolumn{1}{c|}{$\mathcal{Z}_{\poisoneddata}$} & \multicolumn{1}{c}{$\mathcal{Z}_{\mathcal{D}}$} & \multicolumn{1}{c}{$\mathcal{Z}_{\poisoneddata}$}  \\ 
\midrule
$\sigma=0.1$    & 722.98  & 0.56    & 2272.35   & 0.08  & 25273.61   & 0.004   \\
$\sigma=1$    & 7.35  & 0.45    & 29.05   & 0.62  & 233.35   & 0.15   \\
$\sigma=10$    & 0.07  & 0.16    & 0.82   & 0.74  & 0.84   & 0.20   \\
\bottomrule
\end{tabular}
\caption{Test statistics $\mathcal{Z}_\data$ for different values of $\sigma$ and $T$.}
    \label{tab:input_attack_appendix}
\end{table}

In \cref{tab:input_attack_appendix} are shown the test statistics $\mathcal{Z}_\data$ and $\mathcal{Z}_\poisoneddata$ for different values of $T$ and $\sigma$ when $u_t\sim\mathcal{N}(0,1)$ and $\Delta u_t=-u_t+a_t$, with $a_t\sim\mathcal{N}(0,1)$. Clearly, the presence of poisoning is  more detectable for small values of the variance of the process noise, while for higher values we need a significantly larger number of samples to be able to detect the effect of poisoning (for $\sigma=10$ we require $T=1000$ samples to detect a difference). The confidence regions for this specific attack depict this effect, and are shown in \cref{fig:input_attack_appendix}.
\newpage
\begin{figure}[h]
    \centering
    \includegraphics[width=\linewidth]{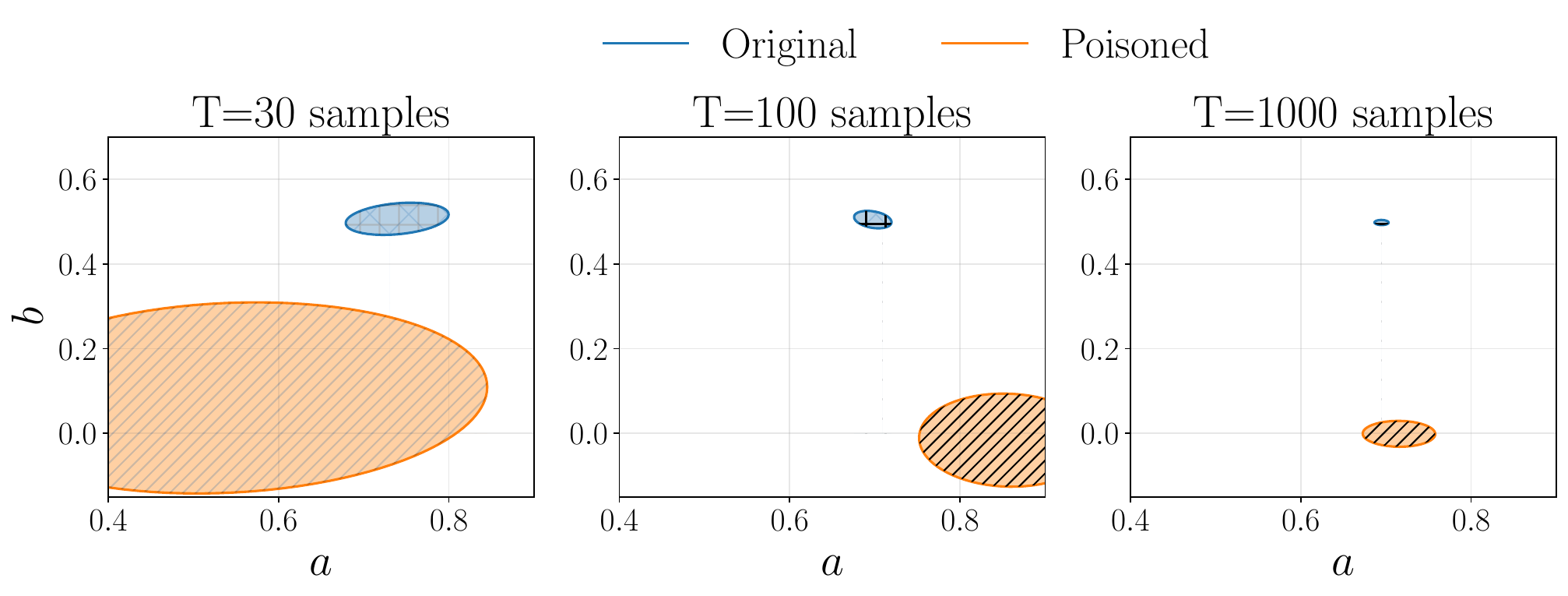}
    \includegraphics[width=\linewidth]{images/input_attack/input_poisoning_1.pdf}
    \includegraphics[width=\linewidth]{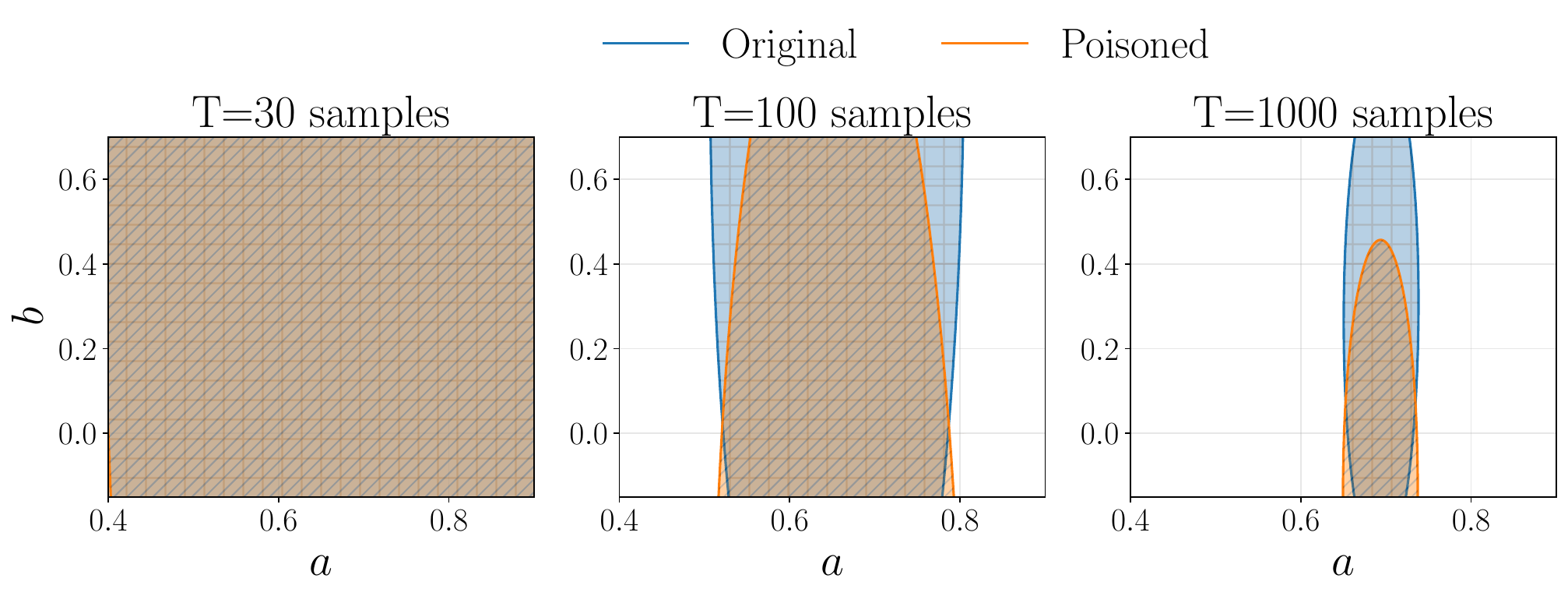}
    \caption{95\% confidence regions for $(a,b)$ and different values of the noise variance. From top to bottom: $\sigma=0.1$, $\sigma=1$, $\sigma=10$.}
    \label{fig:input_attack_appendix}
\end{figure}
\newpage
\subsubsection{Residuals maximization attack}
In \cref{example:residuals_variance} we consider a system described by the following transfer function
\begin{equation}
    P(s) = \frac{0.28261s+0.50666}{s^4-1.41833s^3+1.58939s^2-1.31608s+0.88642},
\end{equation}
sampled with sampling time $\Delta T=0.05s$. From lemma \ref{lemma:MSE_residuals_under_attack} we know that  $\|\tilde R\|_{\textnormal{F}}^2 = \trace(W_- F_TW_-^\top) +\trace(\Delta W_- F_T\Delta W_-^\top) +2\trace( W_- F_T \Delta W_-^\top)$, where $F_T=(I_T - \tilde M)$ and $\tilde M=\begin{bmatrix}
    \tilde X_- \\ \tilde U_- 
    \end{bmatrix}^\dagger \begin{bmatrix}
    \tilde X_- \\ \tilde U_- 
    \end{bmatrix}$ is a $T\times T$ idempotent matrix. Since the first term $ \trace(W_- F_TW_-^\top)$ does not depend on the attack signal, for our purposes we only need to consider the latter two terms. Therefore, if the attacker wants  to maximize the norm of the residuals, she simply needs to maximize 
$
   \trace(\Delta W_- F_T\Delta W_-^\top) +2\trace( W_- F_T \Delta W_-^\top).
$
 Since the true noise sequence $W_-$ is not available, we approximate it using the unpoisoned residuals $R$. Similarly, in the computation of $\Delta W_-$ we need to know what is the true parameter $(A_\star,B_\star)$. As an approximation, we replace this quantity by the unpoisoned estimate $(A_\LS, B_\LS)$ . Therefore, the objective of the attacker is to maximize
\begin{equation}
    \trace(\Delta \tilde W_- F_T\Delta \tilde W_-^\top) +2\trace( R F_T \Delta \tilde W_-^\top),
\end{equation}
where $\Delta \tilde W_-=\Delta X_- -A_\LS\Delta X_-- B_\LS \Delta U_-. $ This last term is clearly convex in $(\Delta X, \Delta U_-)$, consequently maximizing it is a concave problem, where the solution is attained at the boundary of the feasible convex set, defined by the inequalities $ \|\Delta X\|_{\textnormal{F}} \leq \delta \| X\|_{\textnormal{F}}$ and $ \|\Delta U_-\|_{\textnormal{F}} \leq \delta \| U_-\|_{\textnormal{F}}$. Note that the total number of parameters to optimize is $Tm + (T+1)n$.
\begin{figure}[t]
    \centering
    \includegraphics[width=\linewidth]{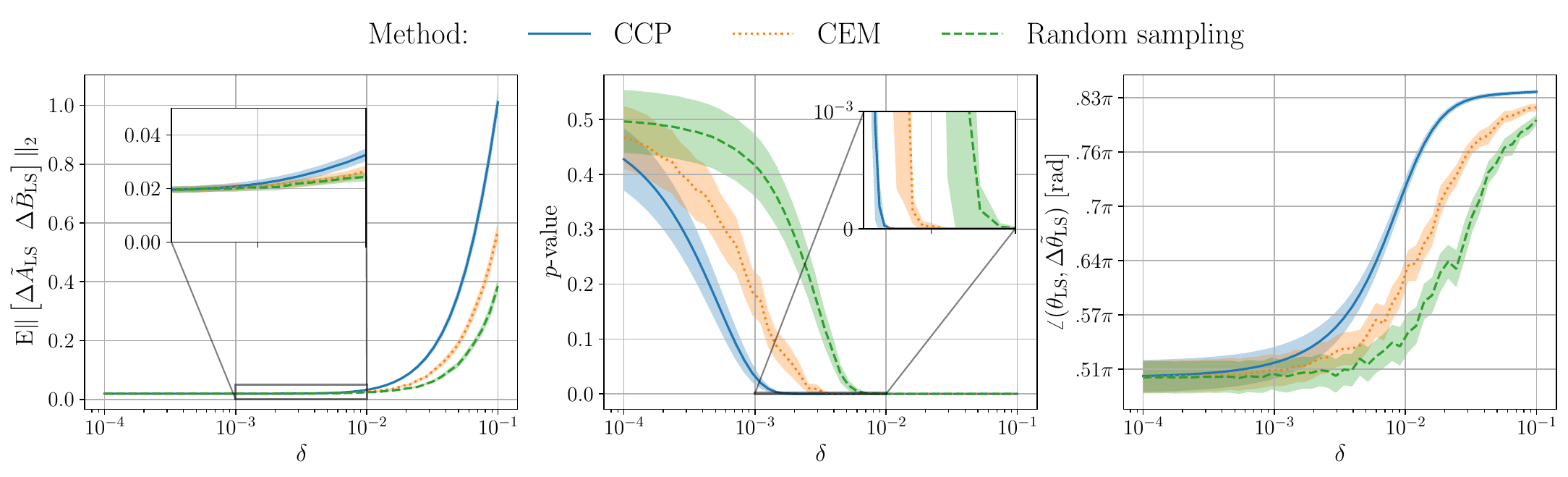}
    \caption{Untargeted attack that maximizes the MSE. On the left is shown the impact on the LS estimate; in the middle  the probability under the null hypothesis of having observed $\|\tilde R\|_{\textnormal{F}}^2$; on the right the angle between the unpoisoned estimate $\theta_\LS$ and the poisoned error of the estimate $\Delta\tilde\theta_\LS$. Shadowed areas depict 95\% confidence intervals over 100 runs.}
\end{figure}

To solve the problem we considered different techniques:
\begin{itemize}
    \item \textit{Approximately exact solution}: one way to solve the problem is to note that the problem can be rewritten as a difference of convex function. Therefore we can use convex-concave programming. In particular, we made use of the DCCP library \cite{shen2016disciplined} with the MOSEK solver.
    \item \textit{Cross-entropy method (CEM)}: the cross-entropy method \cite{de2005tutorial} can be used to solve the problem using samples generated from a Gaussian distribution (see also \cref{algo:cross_entropy_method}). We used the following parameters for \cref{algo:cross_entropy_method}: $\mu_0=0, \Sigma_0=I_{k}, M=500, N=50, \rho=0.1, \varepsilon=1.5,\alpha=0.5$ and $\lambda = 0.001$ (where $k=Tm + (T+1)n$ is the number of parameters to optimize).
    \item \textit{Random sampling from a Gaussian distribution}: this is similar to the cross-entropy method, we sampled $100$ random points according to $\mathcal{N}(0,I_k)$ (with $k=Tm+(T+1)n$) and re-scaled the covariance matrix to guarantee the constraints on the norm of the poisoning signals. We evaluated the MSE of each sample and chose the one that achieved the largest error.
\end{itemize}
In the simulations we considered a zero-mean Gaussian process noise with variance $\sigma_w=0.1$. The victim has collected $T=200$ samples using an i.i.d. control signal  $u_t\sim\mathcal{N}(0,1)$.

\begin{algorithm}[t] 
\caption{\textsc{Cross-Entropy Method -  Gaussian Distribution}}
\begin{algorithmic}[1]\label{algo:cross_entropy_method}
    \REQUIRE Initial mean  and covariance $\mu_0\in \mathbb{R}^k$ and $\Sigma_0\in \mathbb{R}^{k\times k}$; number of points $N$ and elite points $N_e$; maximum number of iterations $M$; function to evaluate $S$; smoothing factor $\alpha$; covariance threshold $\varepsilon$; regularizer $\lambda$.
    \smallskip
    \FOR{$t=0,\dots,M-1$ iterations}
        \STATE \textbf{Sampling.} Generate $N$ samples $\left\{x_1^{(t)},\dots, x_N^{(t)}\right\}$ from the density $\mathcal{N}(\mu_t,\Sigma_t)$.
        \STATE \textbf{Evaluation.} Evaluate the performance $ S(x_i^{(t)})$ of each sample $x_i^{(t)}$, and let $\mathcal{I}$ be the set of $N_e$ best performing samples.
        \STATE \textbf{Update.} {Compute mean and covariance using the set of best performing samples
        \begin{align*}
            \hat \mu_t &= \frac{1}{N_e}\sum_{j\in \mathcal{I}} x_{j}^{(t)},\\
            \hat \Sigma_t &= \frac{1}{N_e}\sum_{j\in \mathcal{I}} \left(x_{j}^{(t)} - \hat\mu_t\right)\left(x_{j}^{(t)} - \hat\mu_t\right)^\top,
        \end{align*}
        and update the parameters according to a smooth update
        \[
            \mu_{t+1} = (1-\alpha) \mu_t + \alpha \hat \mu_t,\quad \Sigma_{t+1} = (1-\alpha)\Sigma_t + \alpha (\hat \Sigma_t + \lambda I_k).
        \]
        \IF{$\sigma_{\text{max}}(\Sigma_{t+1})\leq \varepsilon$}
            \STATE Break the loop.
        \ENDIF
        }
    \ENDFOR
    \RETURN Overall best solution $x_{\text{best}}$ generated by the algorithm.
\end{algorithmic}
\end{algorithm}

\subsubsection{Stealthy attack}
We propose the following stealthy attack, computed by solving the following optimization problem:
\begin{equation}
        \begin{aligned}
            \max_{\Delta U_{-}, \Delta X} \quad & \norm{\begin{bmatrix}\Delta \tilde A_\LS & \Delta  \tilde B_\LS \end{bmatrix}}_2 
            \textrm{ s.t. }  g_i(\data, \Delta U, \Delta X)\leq \delta_i,\quad i=0,2,\dots,s+3,
        \end{aligned}
\end{equation}
where
\begin{align}
    g_0 &= \|\Delta X\|_\textnormal{F}/\|X\|_\textnormal{F},\\
    g_1&= \|\Delta U\|_\textnormal{F}/\|U\|_\textnormal{F},\\
    g_2&=\left|1-\|\tilde R\|_{\textnormal{F}}^2/\| R\|_{\textnormal{F}}^2 \right|,\\
    g_3 &= |1-Z_\poisoneddata /Z_\data|,\\
    g_{3+\tau} &=\left | 1- \|\tilde C_\tau \tilde C_0^{-1}\|_{\textnormal{F}}^2 /\| C_\tau  C_0^{-1}\|_{\textnormal{F}}^2 \right |,\quad \tau=1,\dots, s,
\end{align}
for some $s<T-1$.  As previously explained, the constraints limit the relative change of each quantity and can be made more granular, for example, by limiting the norm of the signals at each time step $t$.

\paragraph{Simulation settings.} We applied the attack on the following continuous-time system
\begin{equation}
    P(s) = \frac{0.28261s+0.50666}{s^4-1.41833s^3+1.58939s^2-1.31608s+0.88642},
\end{equation}
sampled with sampling time $\Delta T=0.05s$. The process noise is zero-mean white noise with standard deviation $\sigma_w=0.1$. Using $10$ different seeds, we collected $10$ different datasets $\data$, each with $T=500$ samples. For each datataset, 10 attacks were computed by solving the above optimization problem using sequential quadratic programming (SLSQP, a quasi-Newton approach). Out of the $10$ best attacks for each dataset, we kept only the attack that maximized the MSE.
To run the simulations, we used a local stationary computer with Ubuntu 20.10, an Intel® Xeon® Silver 4110 Processor (8 cores) and 64GB of ram. On average, it took $7$ hours to compute $10$ attacks on a single dataset $\data$, for $\delta\in\{0.01,0.025,0.05,0.075,0.1\}$. The total simulation time took approximately  $7$ days. For further details, please refer to the repository.

\paragraph{Results and discussion.} In \cref{fig:appendix_statistics}, \cref{fig:appendix_residuals_correlation} and \cref{fig:appendix_example_poisoning_delta} are shown the results of the simulation. In \cref{fig:appendix_statistics} and \cref{fig:appendix_residuals_correlation} are shown the statistical quantities of interest. In \cref{fig:appendix_statistics} is also shown the boxplot of all leverage scores $h_{ii}$, $i =0,\dots,T-1$ ($h_{ii}$ is the $i$-th element along the diagonal of $\tilde \Psi_-^\dagger \tilde \Psi_-$). We note that for small values of $\delta$, the norm of the residuals $\|\tilde R_t\|_2$, as well as the norm of the normalized correlation matrices $\|\tilde C_\tau \tilde C_0^{-1}\|_\frobenius^2$, are comparable to the unpoisoned case. Interestingly, the norm of the residuals seem to be more affected by the poisoning for large values of $\delta$, which is,  however, not captured by the leverage score in \cref{fig:appendix_statistics}.

From \cref{fig:appendix_residuals_correlation} we also observe that the constraints $g_0$ and $g_1$ for large values of $\delta$ may not be very effective. In fact we note that the average norm of the residuals tend to decrease, while the number of peaks increases. This suggests that an attacker may use more granular constraints, which, however, will make the attack less effective.

However, these results also demonstrate that this attack can severely impact the LS estimate for small values of $\delta$, compared to the residuals attack presented in \cref{example:residuals_variance}. In comparison, for $\delta=0.05$ the attack impact has more than doubled, while  the statistical indicators show no evidence of  anomaly. In addition,  for small values of $\delta$ the residuals visually do not change in a sensible way, and an  analysis of the outliers, based on the concept of leverage, shows  no statistical difference for any values of $\delta$ . Furthermore, for any value of $\delta$ the attack seems to impact more the estimate of $A_\LS$ rather than $B_\LS$,  due to the presence of the constraint $g_3$.

\begin{figure}[h]
    \centering
    \includegraphics[width=0.9\linewidth]{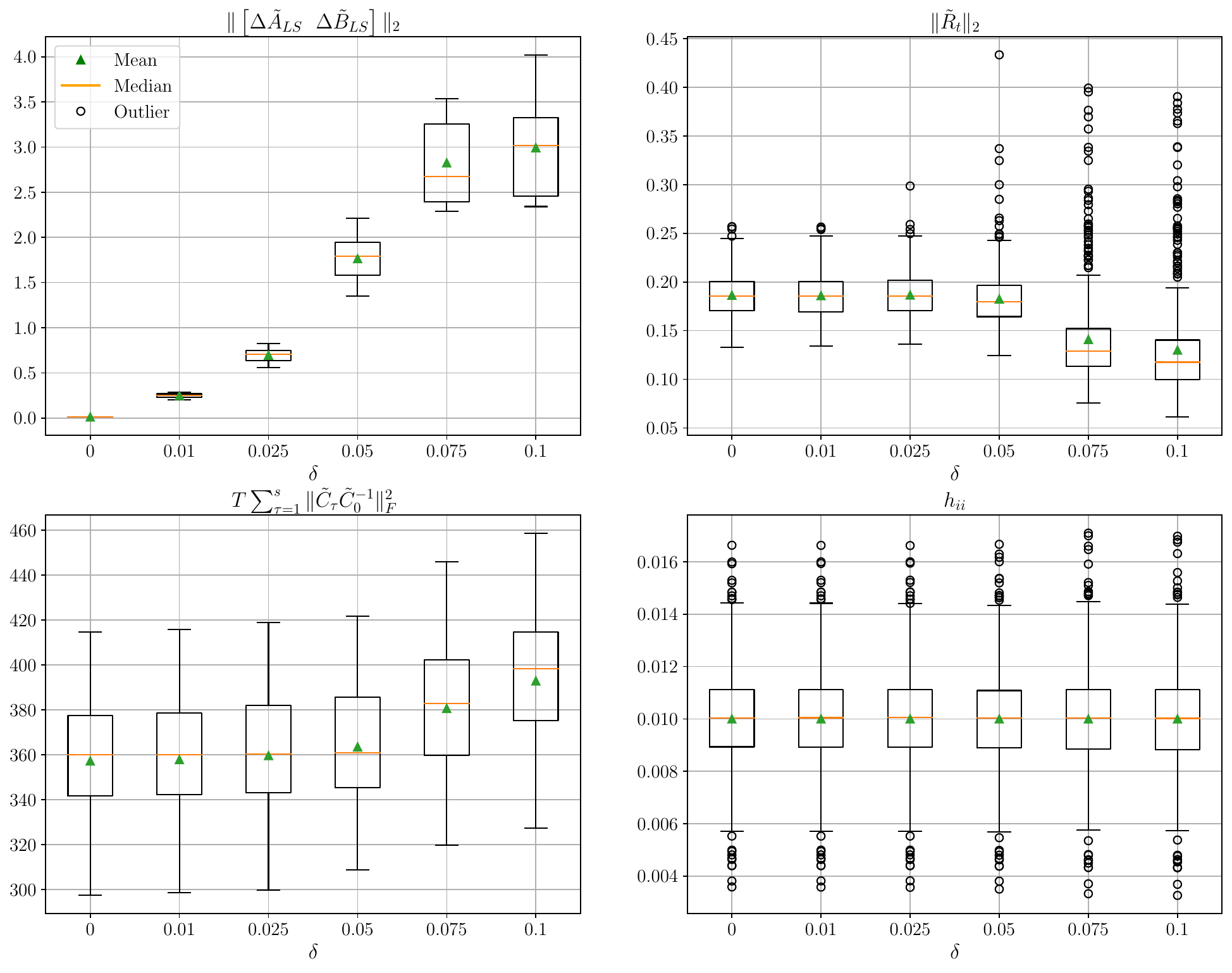}
    \includegraphics[width=0.9\linewidth]{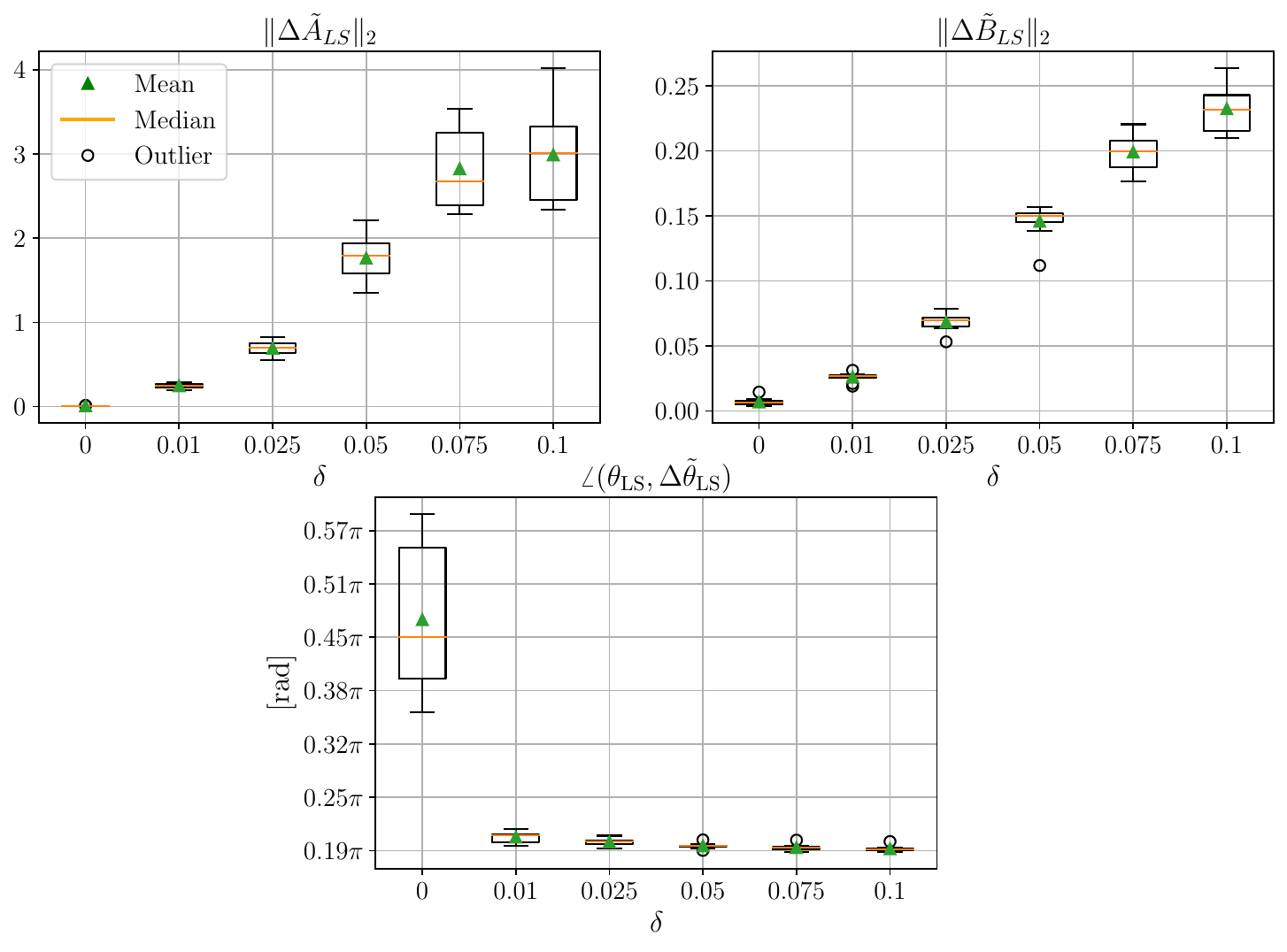}
    
    \caption{Statistics of the attack in \cref{eq:optimal_attack} over $10$ datasets for different values of $\delta$.}
    \label{fig:appendix_statistics}
\end{figure}

\begin{figure}
    \centering
    \includegraphics[width=0.9\linewidth]{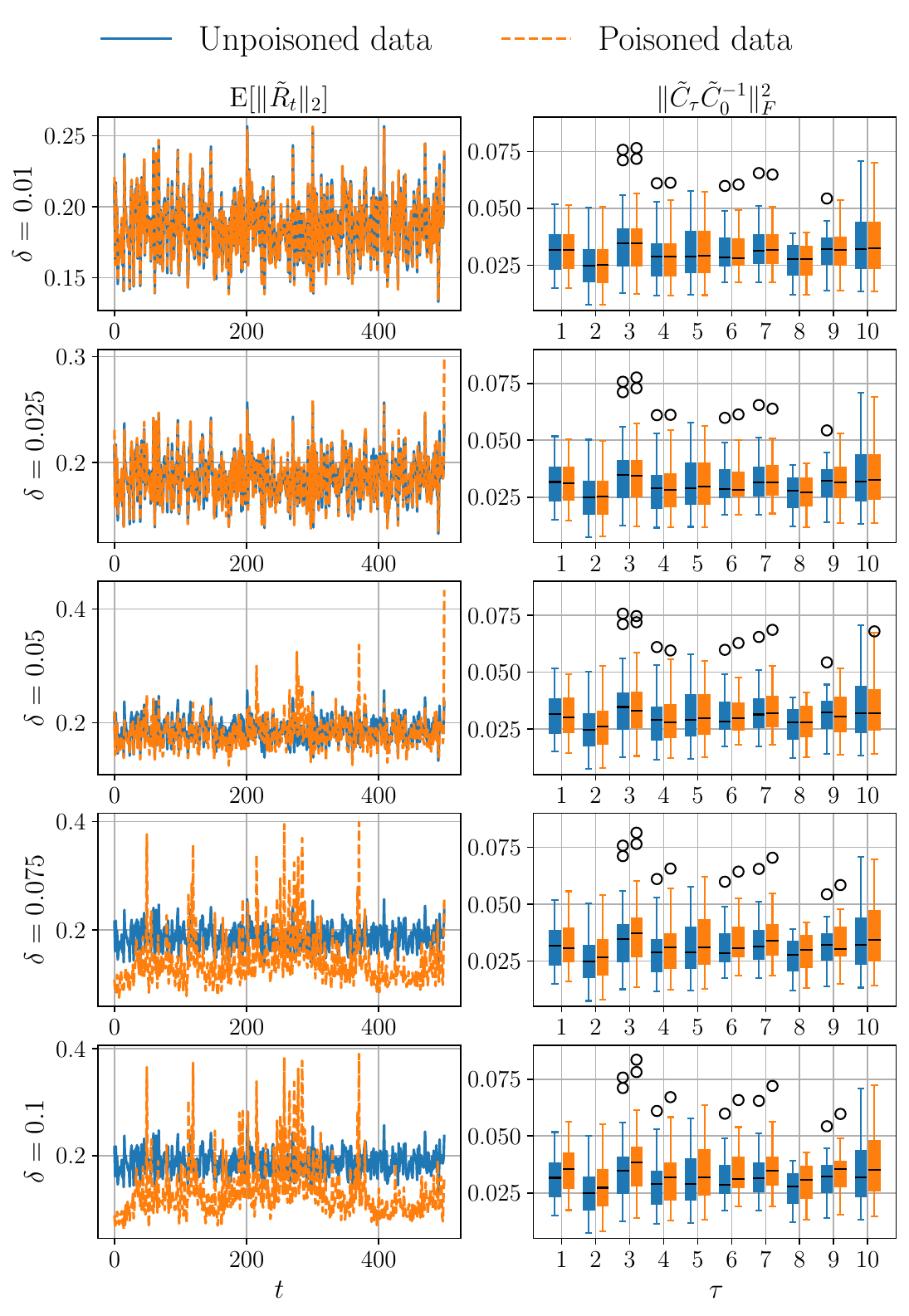}
    \caption{Residuals analysis for one of the dataset. On the left is shown the norm of the residuals when the data has been poisoned according to the attack strategy in \cref{eq:optimal_attack}, for different values of $\delta$. On the right is shown $\|\tilde C_\tau \tilde C_0^{-1}\|_\frobenius^2$ for different values of $\tau$.}
    \label{fig:appendix_residuals_correlation}
\end{figure}

\begin{figure}
    \centering
    \includegraphics[width=0.93\linewidth]{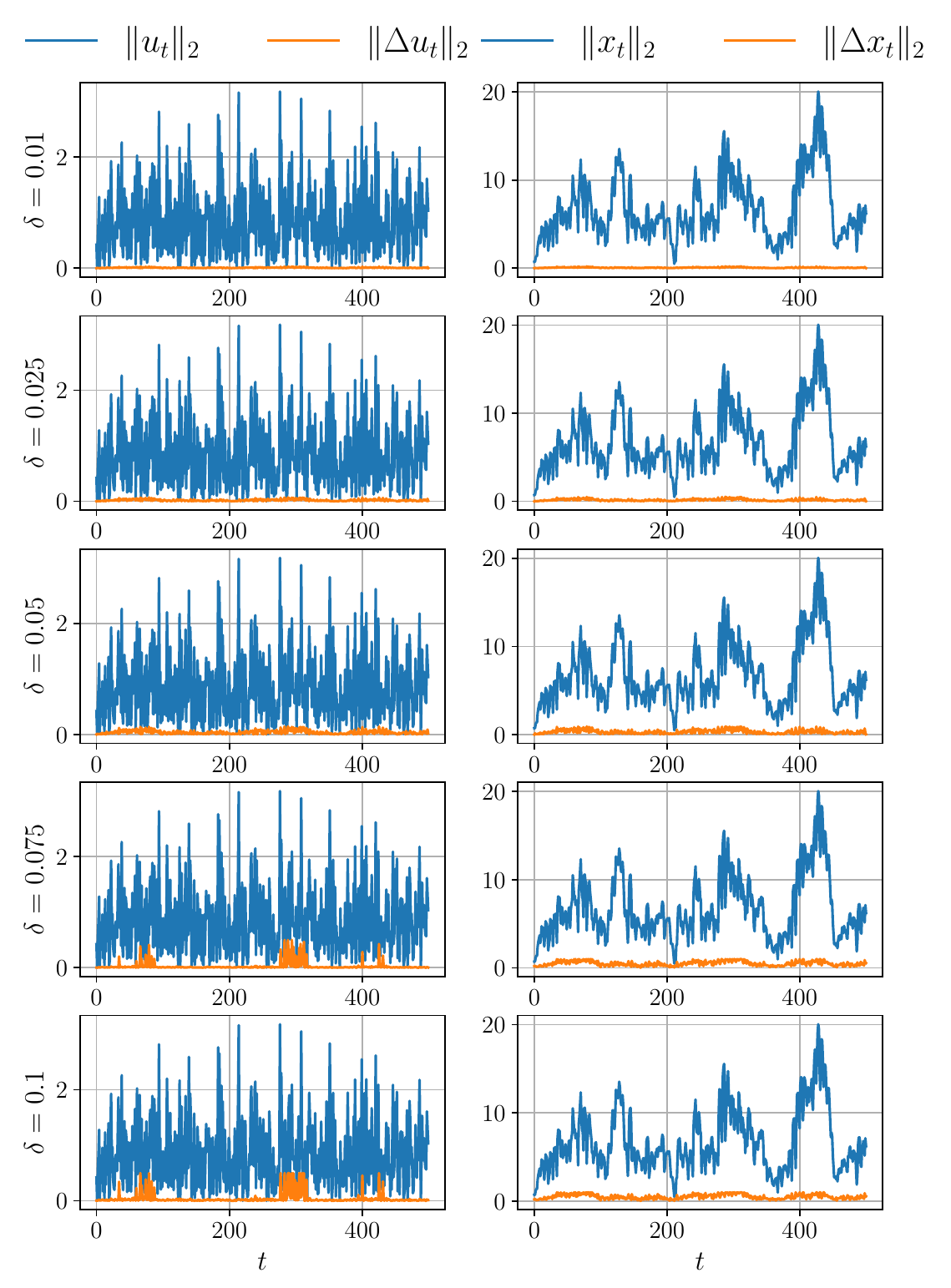}
    \caption{Comparison between the original signals $\{(u_t,x_t)\}_t$ and the poisoning signals $\{(\Delta u_t,\Delta x_t)\}_t$ for one of the dataset (computed by solving \cref{eq:optimal_attack} for different values of $\delta$).}
    \label{fig:appendix_example_poisoning_delta}
\end{figure}
\else\fi
\end{document}